\documentclass{jac}
\usepackage[utf8]{inputenc}

\usepackage{stmaryrd,amsmath,amsthm,amsfonts}
\usepackage{color,tikz}
\usepackage[pdftex,bookmarks=true,bookmarksnumbered=true,pdfpagelayout=SinglePage,colorlinks=false]{hyperref}

\let\dfn\emph
\let\impl\Rightarrow
\newcommand{\resp}[1]{\ (resp. #1)}

\newcommand{\Z}{\mathbb Z}

\newcommand{\N}{\mathbb N}
\newcommand{\M}{\mathbb M}
\newcommand{\Ns}{\N_+}
\newcommand{\az}{{A^\Z}}

\newcommand{\sett}[2]{\left\{\left.#1\vphantom{#2}\right|#2\right\}}
\newcommand{\set}[3]{\sett{#1\in#2}{#3}}
\newcommand{\oo}[2]{\left\rrbracket #1,#2\right\llbracket}
\newcommand{\cc}[2]{\left\llbracket #1,#2\right\rrbracket}
\newcommand{\scc}[2]{_{\cc{#1}{#2}}}
\newcommand{\soo}[2]{_{\oo{#1}{#2}}}
\newcommand{\co}[2]{\left\llbracket #1,#2\right\llbracket}
\newcommand{\sco}[2]{_{\co{#1}{#2}}}
\newcommand{\kaprx}[1]{\mathcal A_{#1}}

\newcommand{\dinf}[1]{\vphantom{#1}^\infty{#1}^\infty}
\newcommand{\lang}{\mathcal L}
\newcommand{\ie}{\textit{i.e.}\ }

\newcommand{\restr}[1]{_{\left|#1\right.}}
\newcommand{\soit}[1]{\left|\everymath{\displaystyle\everymath{}}\begin{array}{ll}#1\end{array}\right.}
\newcommand{\appl}[5]{\begin{array}{rrcl}#1:&#2&\to&#3\\
&#4&\mapsto&\displaystyle#5\end{array}}

\theoremstyle{definition}

\begin{document}

\title{Clandestine Simulations in Cellular Automata}

\author[lab1]{P. Guillon}{P. Guillon}
\address[lab1]{Department of Mathematics, University of Turku,\newline 20014 Turku, Finland}
\email[P. Guillon]{piegui@utu.fi}

\author[lab2]{P.-E. Meunier}{P.-E. Meunier}
\address[lab2]{LAMA (CNRS, Universit\'e de Savoie),\newline Campus Scientifique 73376 Le Bourget-du-Lac Cedex, France}
\email[P.-E. Meunier]{pierreetienne.meunier@univ-savoie.fr}
\email[G. Theyssier]{guillaume.theyssier@univ-savoie.fr} 

\author[lab2]{G. Theyssier}{G. Theyssier}

\thanks{Research partially supported by projects CONICYT-ANILLO ACT88 (Chile), FONDECYT 1090156 (Chile), ANR EMC NT09 555297 (France) and Academy of Finland project 131558.}


\begin{abstract}
  \noindent This paper studies two kinds of simulation between cellular
  automata: simulations based on factor and simulations based on
  sub-automaton. We show that these two kinds of simulation behave in
  two opposite ways with respect to the complexity of attractors and
  factor subshifts. On the one hand, the factor simulation preserves the
  complexity of limits sets or column factors (the simulator CA must
  have a higher complexity than the simulated CA). On the other hand,
  we show that any CA is the sub-automaton of some CA with a simple
  limit set (NL-recognizable) and the sub-automaton of some CA with
  a simple column factor (finite type). As a corollary, we get
  intrinsically universal CA with simple limit sets or simple column
  factors. Hence we are able to 'hide' the simulation power of any CA
  under simple dynamical indicators.
\end{abstract}

\maketitle

\section*{Introduction}

Since the introduction of the model in the 40s, cellular automata have
been studied both as dynamical systems and as a computational model. In both aspects, they can show very complex behaviors, be it through their topological dynamics \cite{kurka} or through their ability to compute \cite{vneumann2,univind}. As such, they constitute a good model to tackle one of the major question of natural computing: what kind of dynamical behavior allows to support computation, and reciprocally, what does the ability to compute imply on the dynamical behavior of a system?

In this paper, we focus on asymptotic dynamics of cellular automata (notion of limit set \cite{cpyu}), and on their unprecise observation (notion of column factors \cite{classif}). Intuitively, the limit set is the set of configuration that can appear arbitrarily far in the evolution of the system, and the column factor is the set of sequences of states that a cell (or group of cells) can take in a valid orbit of the system. These notions have been intensively studied in the literature as indicators of the dynamical complexity of cellular automata \cite{langlim2,rice,oexp,utrace}. 

Concerning limit sets of cellular automata, the initial (wrong) intuition was that a universal CA should necessarily have a non-recursive limit set (such a statement appears in \cite{cpyu}). Later, a Turing-complete CA with a simple limit set was constructed in \cite{limuniv}. This result gives a first hint concerning the absence of correlation between the complexity of the limit set and the ability to handle computations. However, the construction goes through a slow simulation of two-register machines where registers are encoded in unary. Therefore, this results is \textit{extrinsic} to the model of cellular automata and shows only that any behavior of register machines can be embedded into a cellular automaton with a simple limit set.

Here, we aim at exploring \textit{intrinsic} versions of the same question: what can be the limit set complexity of a CA able to simulate any other CA? More precisely, we consider two flavors of simulations: one using factor (continuous projection) of a simulator CA onto the simulated CA, and the other using the local uniform injection (sub-automaton) of a simulated CA into a simulator CA. Such simulation relations were extensively studied in \cite{bulk2}, and the second flavor is the one giving rise to the notion of intrinsic universality \cite{ollinger,bulk2}. The main point of the present paper is that factor simulations preserve limit set complexity whereas sub-automaton simulations do not. We also show that the same phenomenon appears for the complexity of column factors. Our main results are two embedding theorems showing that there exist an intrinsically universal CA with column factors of finite type and an intrinsically universal CA with an NL-recognizable limit set (that is, its limit language can be recognized by a non-deterministic Turing machine in logarithmic space). The second result solves an open problem of \cite{theyssier}.

\section{Definitions}

Let us note $\Ns=\N\setminus\{0\}$. If $i,j\in\Z$, we note $\cc ij$ the interval of integers $k$ such that $i\le k\le j$, $\co ij=\cc i{j-1}$, $\oo ij=\cc{i+1}{j-1}$ and so on.
Consider a fixed finite \dfn{alphabet} $A$.
If $x\in A^\Z$ and $\cc ij\subset\Z$, then we note $x\scc ij\in A^{j-i+1}$ the \dfn{pattern} corresponding to the sequence of letters $x_i\ldots x_j$ (and similarly for other kinds of intervals).

If $U\subset A^k$ for some $k\in\N$ and $i\in\Z$, the \dfn{cylinder} $[U]_i$ will denote the set of configurations $x\in A^\Z$ such that $x\sco i{i+k}\in U$. We also note $[U]=[U]_0$.
If $a\in A$, let $\dinf a$ be the configuration $x\in\az$ such that $x_i=a$ for any $i\in\Z$.

A \dfn{dynamical system} is a pair $(X,F)$ where $X$ is a compact space and $F$ is a continuous self-map of $X$. We can then study iterations $F^t$ for any \dfn{time step} $t\in\N$.

Let $\M$ stand either for $\N$ or for $\Z$. The \dfn{shift} map $\sigma$ is defined for any $z\in A^\M$ and any $i\in\M$ by $\sigma(x)_i=x_{i+1}$.
A \dfn{subshift} is a dynamical system $(\Sigma,\sigma)$, or simply $\Sigma$, that is a subset of $A^\M$ which is closed under the usual Tychonoff topology and invariant by the shift map.
Equivalently, it is the set $\set z{A^\M}{\forall\cc ij\subset\M,z\scc ij\notin\mathcal F}$ of infinite words that avoid the finite patterns from some given family $\mathcal F$.
If this family can be taken finite, then $\Sigma$ is called a subshift of \dfn{finite type} (SFT). If it can be taken among words of length $k\in\Ns$, it has \dfn{order} $k$.
The \dfn{language} $\lang(\Sigma)$ of a subshift $\Sigma$ is the set $\sett{z\scc ij}{z\in\Sigma\text{ and }\cc ij\subset\M}$ of patterns appearing in the infinite words of $\Sigma$.
If this language is regular, then we say that $\Sigma$ is a \dfn{sofic} subshift. Equivalently thanks to the Weiss theorem \cite{weiss}, it is obtained from an SFT by a letter-to-letter projection.

A \dfn{cellular automaton} (CA) is a dynamical system $(\az,F)$ which commutes with the shift, \ie $\sigma F=F\sigma$.
Equivalently, thanks to the Curtis-Hedlund-Lyndon theorem, it is obtained by some \dfn{local map} $f:A^{2r+1}\to A$ for some \dfn{radius} $r\in\N$, \ie for any $x\in\az$ and $i\in\Z$, $F(x)_i=f(x\scc{i-r}{i+r})$.
$F$ admits a \dfn{spreading} state $0\in A$ if it admits a local rule $f:A^{2r+1}\to A$ with $r\in\N$ and $f(u)=0$ whenever there exists $i\in\cc0{2r}$ such that $u_i=0$.

An SFT $\Sigma$ is \dfn{irreducible} if for any two words $u,v\in\lang(\Sigma)$, there exists a word $w$ such that $uwv\in\lang(\Sigma)$.
CA can actually be applied in a very natural way to these subshifts: a \dfn{partial CA} will be a dynamical system over some irreducible SFT that commutes with the shift. It is known from \cite{edensof} that any injective partial CA is bijective and reversible (the inverse is also a partial CA).


\subsection{Simulations}

The central notion studied in this paper is that of simulation between dynamical
systems and more specifically cellular automata. We will distinguish two
families of simulation relations based on the following notions.

\begin{defi}[Factors and sub-systems]\leavevmode 
  \begin{itemize}
  \item A \dfn{factor map} between two dynamical systems $(X,F)$ and $(Y,G)$ is a
    continuous onto map $\Phi:X\to Y$ such that $\Phi F=G\Phi$.
    In that case, we say that $G$ is a \dfn{factor} of $F$. 
  \item A \dfn{sub-system} of a dynamical system $(X,F)$ is a dynamical system of
    the form $(Y,G)$ where ${Y\subseteq X}$ ($Y$ is closed) and ${F(Y)\subseteq Y}$.
  \end{itemize}
\end{defi}

Note that a factor map $\Phi$ between two subshifts $\Sigma$ and $\Gamma$ respect the Curtis-Hedlund-Lyndon theorem: there exist a \dfn{radius} $r\in\N$ and a \dfn{local rule} $\phi:\lang_{2r+1}(\Sigma)$ such that $\Phi(x)_i=\phi(x\scc{i-r}{i+r})$ for any $x\in\Sigma$ and any $i\in\Z$. We say that the factor map between two subshifts is a \dfn{coloring} if the radius can be taken $r=0$.

If $(X,F)$ and $(Y,G)$ are cellular automata, we say that \emph{$(X,F)$ simulates $(Y,G)$ by factor} if there is a factor map $\Phi$ from $(X,F)$ onto $(Y,G)$ which is also a factor map from $(X,\sigma_X)$ onto $(Y,\sigma_Y)$.
Besides, when $(X,F)$ is a cellular automaton and $Y$ is a full-shift included in $X$, then $(Y,F)$ is a \dfn{sub-automaton} of $(X,F)$. 

These two relations (factor and sub-system) are restrictive since they don't allow entropy to increase: a factor or a sub-system of a given system has always a lower entropy. As a consequence, they don't support universality (existence of a system able to simulate any other) since it is not difficult to build systems of arbitrarily large entropy. Hence, in the literature, other ingredients were introduced to obtain richer notions of simulation: for instance, in cellular automata, operations of space and time rescaling added to the notion of sub-automaton lead to a notion of simulation supporting universality \cite{bulk1,bulk2}. 
 
In this paper, such kind of spatio-temporal transformations are not considered explicitly for two reasons:
\begin{itemize}
\item our results about factor simulation (Section~\ref{s:facsim}) involve properties
(complexity of the subshifts) which are invariant by space and time
rescaling, hence the results still hold when considering rescaling as part of the simulation relation;
\item our results about sub-system simulation (Section~\ref{s:subsim}) are of the form
"any CA is the sub-automaton of a CA with some given property", which is
actually the most general we can get, and remain true when replacing
"sub-automaton" by more general notions of sub-system (such as
sub-system, or simulated system in the sense of \cite{bulk2}).
\end{itemize}

\subsection{Complexity of limit sets and column factors}

Many different points of view have been adopted to study the complexity of CA. We here use symbolic dynamics, and consider the complexity, as subshifts, of the limit set on the one hand, or the column factors on the other hand, as representing the actual complexity of the CA.

The \dfn{limit set} of the dynamical system is the nonempty closed subset $\Omega_F=\bigcap_{t\in\N}F^t(X)$. Its \dfn{limit system} is the maximal surjective subsystem, $(\Omega_F,F)$. It basically represent the asymptotic dynamics of the system.

With respect to CA limit sets, it is not difficult to see that the corresponding language is always corecursively enumerable (it is an \emph{effective} subshift). However, there are known examples of non-recursive ones \cite{langlim2}. Moreover, simple additional remarks \cite{revlim,kurka,nasu} give the following hierarchy:\\
$F$ injective $\impl$ $F\restr{\Omega_F}$ injective $\impl$ $\Omega_F$ is an SFT $\impl$ $F$ is stable $\impl$ $\Omega_F$ is sofic $\impl$ \ldots\

The \dfn{column factor} of a dynamical system $(\az,F)$ upon interval $\cc ij\subset\Z$ is the set $\tau_F^{\cc ij}=T_F^{\cc ij}(\az)$, where $\appl{T_F^{\cc ij}}\az{(A^{\cc ij})^\N}x{(F^t(x)\scc ij)_{t\in\N}}$. It is a factor subshift since $\sigma T_F^{\cc ij}=T_F^{\cc ij}F$. It can represent an observation of the system made by a measuring device with a finite precision (that cannot see cells which are far away). It can be seen that any factor subshift is essentially a factor of some column factor \cite{classif}.

In the case of CA, shift-invariance allows to consider only the central column factors $\tau_F^{\cc{-k}k}$ for radius $k\in\Ns$.
It is known \cite{notes} that these CA column factors always have a context-sensitive language; they may actually be strictly context-sensitive. In \cite{classif}, strong links with topological notions are stated: finite column factors is equivalent to equicontinuity, SFT column factors imply the shadowing property which in turn imply sofic column factors.

\section{Factor simulations}\label{s:facsim}

The two hierarchies that we have just defined, based on subshift classifications, are then very robust to factor simulations, as we can see in this section.
Taking the vocabulary of order theory, we say that a class $\mathcal C$ of systems is \emph{an ideal for factor simulation} if, whenever $(X,F)$ is a factor of $(Y,G)$, we have: ${(Y,G)\in\mathcal C\Rightarrow(X,F)\in\mathcal C}$.

\begin{prop}\label{p:limsim}
If $\Phi$ is a factor map from $(X,F)$ onto $(Y,G)$, then $\Omega_G=\Phi(\Omega_F)$.
\end{prop}
\begin{proof}
For any $t\in\N$, $\Phi F^t(X)=G^t(Y)$.
First note that $\Phi(\bigcap_{t\in\N}F^t(X))$ is included in $\bigcap_{t\in\N}\Phi F^t(X)=\bigcap_{t\in\N}G^t(Y)$.
Conversely, let $y\in\bigcap_{t\in\N}G^t(Y)$ and, for $t\in\N$, $X_t=\Phi^{-1}(y)\cap F^t(X)$. Note that $X_t$ is closed (since $\Phi$ and $F$ are continuous, and $X$ is compact) and nonempty (since $\Phi$ is onto). By compactness, the intersection $\bigcap_{t\in\N}X_t=\Phi^{-1}(y)\cap\Omega_F$ is not empty, \ie $y\in\Phi(\Omega_F)$.
We have proven that $\Phi(\Omega_F)=\Omega_G$.
\end{proof}

Moreover, we can note that $\Omega_{F^n}=\Omega_F$, since if $n\in\Ns$, then $F^{nt}(X)=\bigcap_{(n-1)t<j\le nt}F^j(X)$.
In the following we will no longer mention time and space rescaling, which does not essentially alter the results.

\begin{cor}
 The class of stable systems is an ideal for factor simulation.
\end{cor}
\begin{proof}
 If $\Phi$ is a factor map from $(X,F)$ onto $(Y,G)$, and $\Omega_F=F^t(X)$ for some $t\in\N$, then $\Omega_G=\Phi(\Omega_F)=\Phi F^t(X)=G^t(X)$.
\end{proof}

We can also derive the two following results.
\begin{prop}
 The class of CA whose limit set\resp{factor subshifts} is sofic\resp{context-free, context-sensitive,
recursive} is an ideal for factor simulation.
\end{prop}
\begin{proof}
It is known that each of the corresponding classes of subshifts is preserved by factor maps.
Proposition~\ref{p:limsim} states that the limit set of the simulated CA is a factor of the limit set of the simulating CA.\\
Moreover, note, thanks to the transitivity of the notion of factor, that a factor subshift of the simulated CA is also a factor subshift of the simulating CA.
\end{proof}


The following slightly generalizes a result in \cite{theyssier}.
\begin{prop}\label{p:injsim}
The class of reversible partial CA is an ideal for coloring.
\end{prop}
\begin{proof}
Let $\Phi$ be a factor map based on some radius-$0$ local rule $\phi:A\to B$, from a partial CA $(\Sigma,F)$ onto another one $(\Gamma,G)$, where $\Sigma\subset\az$ and $\Gamma\subset B^\Z$ and there exists some partial CA $F^{-1}:\Sigma\to\Sigma$ such that $FF^{-1}$ is the identity.
By surjectivity of $\Phi$, for any letter $b\in B$, there exists a letter, denoted abusively $\phi^{-1}(b)\in A$ such that $\phi(\phi^{-1}(b))=b$. If $\Phi^{-1}$ represents the parallel application of $\phi^{-1}$ to $B^\Z$, we obtain that $\Phi\Phi^{-1}$ is the identity of $\Gamma$.
We can define the map $G^{-1}=\Phi F^{-1}\Phi^{-1}$, in such a way that $GG^{-1}$ is the identity of $\Gamma$. $G^{-1}$ is an injective partial CA, hence it is reversible, and it is the actual inverse map of $G$.
\end{proof}
As far as we know though, it is unknown whether the class of injective CA is still an ideal for factor simulation.
\begin{cor}
The class of CA which are reversible over the limit set is an ideal for coloring.
\end{cor}
\begin{proof}
Consider such a CA.
From \cite{revlim}, it is stable and its limit set is an irreducible SFT.
Now if it is linked to some other CA by some coloring, then by Proposition~\ref{p:limsim}, so are the two limit systems (that are partial CA), which then respect the hypotheses of Proposition~\ref{p:injsim}.
\end{proof}

From the previous, we can see that any universal CA for factor simulation (up to rescaling) must have  a non-recursive limit set and strictly context-sensitive factor subshifts (and of course must not be injective), but the existence of such a CA is still open.

\section{Sub-system simulations}\label{s:subsim}
\label{sec:injsimu}

We will see in this section that, contrary to factor simulations, sub-system simulations allow to hide the complexity of the simulated CA into the simulator CA.

\subsection{Hiding the column factors}

If $G$ is a CA over alphabet $C$ of radius $r>0$, local rule $g$, then let us define the CA $\tilde G$ over alphabet $D=\{-1,0,1\}\times C$ with the same radius $r$ as $G$ and  local rule:
\[\appl{\tilde g}{D^{2r+1}}D{(\varepsilon_{-r},c_{-r})\ldots(\varepsilon_r,c_r)}{\soit{
(\varepsilon_0,c_{\varepsilon_0})&\textrm{if }\varepsilon_0\ne0~;\\(0,g(c_{-r}\ldots c_r))&\textrm{otherwise}.}}\]
Clearly, $G$ is a sub-automaton of $\tilde G$ (up to state renaming) corresponding to the sub-alphabet $\{0\}\times C$.

\begin{thm}\label{t:trsft}
Any cellular automaton $G$ is a sub-automaton of some CA whose column factors are SFT of order $2$.
\end{thm}
\begin{proof}
Let us take a CA $G$ over alphabet $C$, of radius $1$.
Let $\tilde G$ be defined as above over alphabet $\tilde C=\{-1,0,1\}\times C$, $k\in\Ns$ and $\Sigma$ its column factor of width $k$.
Of course, $\Sigma$ is included in its \dfn{$2$-approximation} \[\kaprx2(\Sigma)=\set{z=(z^t)_{t\in\N}}{(\tilde C^k)^\N}{\forall t\in\N,\exists x^t\in[z^t]\cap\tilde G^{-1}([z^{t+1}])}.\]
Let us show that they are actually equal. Let $z=(z^t)_{t\in\N}\in(\tilde C^k)^\N$ be such that for any $t\in\N$ there exists some configuration $x^t\in[z^t]$ with $\tilde G(x^t)\in[z^{t+1}]$, and $x$ defined by:
\[
x_i=\soit{x^0_i&\textrm{if }0\le i<k\\(-1,c)&\textrm{if }i<0\textrm{ and }x^{-i-1}_{-1}=(\varepsilon,c)\\(1,c)&\textrm{if }i\ge k\textrm{ and }x^{i-k}_k=(\varepsilon,c)~.}
\]
An inductive application of the local rule gives that for any $t\in\N$, we have:
\[
\tilde G^t(x)_i=\soit{x^t_i&\textrm{if }0\le i<k\\(-1,c)&\textrm{if }i<0\textrm{ and }x^{t-i-1}_{-1}=(\varepsilon,c)\\(1,c)&\textrm{if }i\ge k\textrm{ and }x^{t+i-k}_k=(\varepsilon,c)~.}
\]
In particular, $z=T_{\tilde G}^k(x)\in\Sigma$.
We have proven that $\Sigma$ is an SFT of order $2$.
If $G$ does not have radius $1$, then it is easy to widen the radius of $\tilde G$ (and increase the speed of the shifts) to get the same result.
\end{proof}

\subsection{Hiding the limit set}

The main result of this section is based on the existence of a firing-squad CA with specific properties expressed by Lemma~\ref{l:fsquad}. We actually refer to the
firing-squad CA defined in \cite{rice}, that we denote by $S$,
and prove additional properties in Section~\ref{sec:fsproof}.
This CA admits a so-called \dfn{firing} state $\gamma$ and a spreading state $\kappa$. 
Let $r_S$ the radius of $S$, $s$ its local rule, $Q$ its state set, and $Q'=Q\setminus\{\kappa,\gamma\}$.
Consider the set $X_S$ of configurations having an infinite history avoiding $\kappa$ and $\gamma$:
\[X_S = \set y{Q^\Z}{\exists (y^t)_{t\in\N}\in(Q^\Z)^\N, y^0=y\text{ and }\forall t\in\Ns,y^t\in Q'\text{ and }S(y^t)=y^{t-1}}~.\]

\begin{lem}\label{l:fsquad}
$S$ satisfies the following:
 \begin{enumerate}
 \item\label{i:fire} $\dinf\gamma\in X_S$;
 \item\label{i:firelim} $\Omega_S\cap[\gamma]\subset\{\kappa,\gamma\}^\Z$;
 \item\label{i:csens} $X_S$ is NL-recognizable.
 \end{enumerate} 
\begin{proof}
  Properties \eqref{i:fire} and \eqref{i:firelim} are proven in \cite{rice}. Property \eqref{i:csens} 
is given by Proposition~\ref{prop:simplefirelimit} below. 
\end{proof}
\end{lem}

This construction allows to state the following theorem, proven at the end of the subsection.
\begin{thm}\label{t:simple}
  Any CA is a sub-automaton of some CA whose limit set is NL-recognizable. 
\end{thm}

By the existence of intrinsically universal CA for a simulation containing the sub-automaton relation (see for instance \cite{univind}) and the transitivity of simulations, we can directly derive the following.
\begin{cor}\label{c:univ}
  There exists an intrinsically universal CA whose limit set is
  NL-recognizable. 
\end{cor}

The idea of the construction is the following: given some CA $F$ over alphabet $A$, we add an extra (firing-squad) component which is able to generate any configuration of $A^\Z$ arbitrarily far in the future. The complexity of the limit set of $F$ is thus completely flooded into the full-shift $A^\Z$. All the technical difficulty is to control the contribution of the additional component to the final limit set. 

Let $F$ be a CA of radius $r_F$, local rule $f$ over alphabet $A$ with a
spreading state $0\in A$. We define a CA $\Delta_{F,S}$ of local rule
$\delta_{F,S}$ defined on alphabet $C=A\sqcup(A\times Q)$ with radius
$r=\max(r_F,r_S)$ by:
\[
\delta_{F,S}:c\mapsto
    \begin{cases}
      f(a_{-r_F}\ldots a_{r_F})&\text{ if }c=a_{-r}\ldots a_r\in A^{2r+1}~;\hfill \text{ (1)}\\
      a_0&\text{ if }c=(a_{-r},\gamma)\ldots(a_r,\gamma)\in(A\times\{\gamma\})^{2r+1}~;\hfill \text{ (2)}\\
      (a_0,s(b_{-r_S}\ldots b_{r_S}))&\text{ if }c=(a_{-r},b_{-r})\ldots(a_r,b_r)\in(A\times Q')^{2r+1}~;\hfill \text{ (3)}\\
      0&\text{ otherwise}.\hfill \text{ (4)}
    \end{cases}
\]
Basically, this CA freezes the first component while applying the firing squad on the second component until some firing state appears, which then frees this second component and starts the application of $F$. When the configuration is not coherent, or when $\kappa$ appears, $0$ begins to spread.
Clearly, $F$ is a sub-automaton of $\Delta_{F,S}$.

The structure of the corresponding limit set will be given by the following lemmas.
\begin{lem}\label{l:sigomeg}
$\az\subset\Omega_{\Delta_{F,S}}$.
\end{lem}
\begin{proof}
Let $x\in\az$. From Point~\ref{i:fire} of Lemma~\ref{l:fsquad}, there is a sequence $(y^t)_{t\in\N}$ with $y^0=\dinf\gamma$ and for any $t\in\Ns$, $y^t\notin\{\gamma,\kappa\}$ and $S(y^t)=y^{t-1}$.
Consider now the configurations $x^t=(x_i,y^{t+1}_i)_{i\in\Z}$ for $t\in\Ns$.
By a quick induction on $t\in\Ns$, we can see that for any cell $i\in\Z$, only case (3) of the local rule is used, and $x^0=\Delta_{F,S}^t(x^t)$. At time $-1$, since the second component of $x^{-1}$ is $\dinf\gamma$, case (2) of the rule is applied in every cell, which gives $x=\Delta_{F,S}(x^{-1})=\Delta_{F,S}^t(x^{-t})$ for any $t\in\Ns$.
\end{proof}

\begin{lem}\label{l:fire}
 Let $x\in\Omega_{\Delta_{F,S}}$ and $i,j\in\Z$ such that $i\ne j$ and $x_i=(a_i,\gamma),x_j=(a_j,b_j)\in A\times Q$.
Then $b_j\in\{\gamma,\kappa\}$.
\end{lem}
\begin{proof}
We can assume, by symmetry, that $i<j$, and for the sake of contradiction that $b_j\notin\{\gamma,\kappa\}$.
Let $(x^t)_{t\in\Z}$ be a biorbit of $x=x^0$, \ie a bisequence of configurations such that $\forall t\in\Z,\Delta_{F,S}(x^t)=x^{t+1}$. By an easy recurrence and the fact that $x_i\in A\times Q$ can only be obtained through case (3) of the rule, we can see that for any $t\in\N$, $x^{-t}\scc{i-rt}{i+rt}$ can be written $(a^{-t}_{i-rt},b^{-t}_{i-rt})\ldots(a^{-t}_{i+rt},b^{-t}_{i+rt})\in(A\times Q)^{1+2rt}$ and $s^t(b^{-t}_{i-r_St}\ldots b^{-t}_{i+r_St})=b_i$; in the same way, $x^{-t}\scc{j-rt}{j+rt}$ can be written $(a^{-t}_{j-rt},b^{-t}_{j-rt})\ldots(a_{j+rt},b_{j+rt})\in(A\times Q)^{1+2rt}$, and $s^t(b^{-t}_{j-r_St}\ldots b^{-t}_{j+r_St})=b_j$.
Then for any $t>\frac{j-i-1}{2r}$, $x^{-t}\scc{i-2rt}{j+2rt}$ is in $(A\times Q)^{j-i-1+4rt}$ and the image $s^t(x^{-t}\scc{i-r_St}{j+r_St})$ contains $b_i$ and $b_j$. In other words, the cylinder $[b_iQ^{j-i-1}b_j]_i$ intersects $S^t(Q^\Z)$ for any $t$, and by compactness intersects $\Omega_S$, which contradicts Point~\ref{i:firelim} of Lemma~\ref{l:fsquad}.
\end{proof}

If $\Sigma\subset\az$ is a subshift and $0\in A$, then we consider the subshift $0\bullet\Sigma\bullet 0=\bigcup_{-\infty\le l\le m\le+\infty}\set x\az{\forall i\notin\oo lm,x_i=0\text{ and }\exists y\in\Sigma,x\soo lm=y\soo lm}$ of configurations or pieces of configurations of $\Sigma$ surrounded by $0$.
\begin{lem}\label{l:ssgamma}
$\Omega_{\Delta_{F,S}}\setminus\az\subset0\bullet(A\times Q)^\Z\bullet0$.
\end{lem}
\begin{proof}
By shift-invariance, it is sufficient to prove that $\Omega_{\Delta_{F,S}}\cap[A\times Q]_0\subset0\bullet(A\times Q)^\Z\bullet0$.  
Let us prove by induction on $n\in\N$ that the patterns of $(A\times Q)(A^{2rn}\setminus\{0^{2rn}\})$ are forbidden in $\Omega_{\Delta_{F,S}}$. The base case is trivial (there are no such patterns). Now suppose it is true for $n\in\N$, and suppose there exists a configuration $x\in[(A\times Q)0^{2rn+k}(A\setminus\{0\})]_0\cap\Omega_{\Delta_{F,S}}$ with $1\le k\le2r$. Consider a preimage $y\in\Omega_{\Delta_{F,S}}$ of $x$. On the one hand, in cell $0$ of $y$, we must have applied case (3), so $y\scc{-r}{+r}\in(A\times Q)^{2r+1}$, and this word does not involve $\gamma$.
On the other hand, if we have applied case (1) in cell $2nr+k+1$ of $y$, then $y\scc{(2n-1)r+k+1}{(2n+1)r+k+1}\in(A\setminus\{0\})^{2r+1}$, but the space between these two neighborhoods is $(2n-1)r+k+1-r-1\le2nr-1$, which contradicts the induction hypothesis.
The other possibility was that we have applied case (2) in cell $2nr+k+1$, which involves a state $\gamma$ among cells of $y\scc{(2n-1)r+k+1}{(2n+1)r+k+1}$, which contradicts Lemma~\ref{l:fire}.
In the limit, and with a symmetric argument on the left, we obtain that all the configurations of $\Omega_{\Delta_{F,S}}\setminus\az$ are in $0\bullet\Sigma\bullet 0$.
\end{proof}

We shall abusively denote $\az\times X_S=\set{(a_i,s_i)_{i\in\Z}}{(A\times Q)^\Z}{(s_i)_{i\in\Z}\in X_S}$.
\begin{lem}\label{l:limfs}
 $\Omega_{\Delta_{F,S}}=\az\cup0\bullet(\az\times X_S)\bullet0$.
\end{lem}
\begin{proof}
Thanks to Lemmas~\ref{l:ssgamma} and \ref{l:sigomeg}, it is enough to prove two inclusions for the configurations $x\in C^\Z$ with $l,m\in\cc{-\infty}{+\infty}$ such that $x\soo lm\in(A\times Q)^{m-l-1}$ and for any $i\notin\oo lm$, $x_i=0$.

First, suppose that $x\in\Omega_{\Delta_{F,S}}$, \ie for any $t\in\Z$, there exists $x^t\in\Delta_{F,S}^t(\{x\})$. By recurrence, we can see that $x^{-t}_i\in A\times Q'$ for all $i\in\oo{l-rt}{m+rt}$ and $t\geq 1$ since states from $A\times Q$ are only produced by case (3) of the rule. Let $w^{-t}\in Q'^{m-l+2rt+1}$ be the projection of $(x^{-t})\soo{l-rt}{m+rt}$ on its second component. Clearly, $w^0$ is in the language of $X_S$. We deduce that $x=x^0\in0\bullet(\az\times X_S)\bullet0$.

Conversely, suppose that $x\in0\bullet(\az\times X_S)\bullet0$, \ie there is a sequence $(y^t)$ with, for $t\geq 1$, $y^t\in Q'^\Z$ and $y^t=S(y^{t+1})$ and, for any $t\in\N$ and any $i\in\oo lm$, $x=(a_i,S^t(y^t)_i)$ for some $a_i\in A$. Now take the configuration $\tilde y^t\in C^\Z$ such that for any $i\notin\oo{l-rt}{m+rt}$, $\tilde y^t_i=0$, and for any $i\in\oo{l-rt}{m+rt}$, $\tilde y^t_i=(b_i,y^t_i)$ with $b_i\in A$, and $b_i=a_i$ if $i\in\oo lm$. By a direct recurrence, for any $j<t$ and any $i\notin\oo{l-rt+rj}{m+rt-rj}$, we have $\Delta_{F,S}^j(\tilde y^j)_i=0$ and for any $i\in\oo{l-rt+rj}{m+rt-rj}$, we have $\Delta_{F,S}^j(\tilde y^j)_i=(b_i,S^j(y^j)_i)$ (since $y_j\in Q'\in\Z$ case 3 of the definition of $\Delta_{F,S}$ applies at position $i$ of $\tilde y^j$). This gives that $\Delta_{F,S}^t(y)=x$.
We have proven that $\Omega_{\Delta_{F,S}}\cap0\bullet(A\times Q)^\Z\bullet0=0\bullet(\az\times X_S)\bullet0$.
\end{proof}
\begin{cor}\label{c:csens}
 $\Omega_{\Delta_{F,S}}$ has an NL-recognizable language.
\end{cor}
\begin{proof}
 From Lemma~\ref{l:limfs} and Point~\ref{i:csens} of Lemma~\ref{l:fsquad}, the language of the limit set is the finite boolean combination of finite concatenation of NL-recognizable languages.
\end{proof}

\begin{proof}[Proof of Theorem~\ref{t:simple}]
Let $F$ be a CA on some alphabet $A$. We can artificially add some spreading
state $0\notin A$ to build a CA $\tilde F$ on alphabet $A\sqcup\{0\}$ which
admits $F$ as a sub-automaton. Now we have seen that $\tilde F$ is a sub-automaton
of $\Delta_{\tilde F,S}$. From Corollary
\ref{c:csens}, the corresponding limit set has an NL-recognizable language.
\end{proof}






\section{Analysis of a firing-squad CA}
\label{sec:fsproof}

More precise proofs can be found in the appendix.

Let $S$ be the firing-squad CA defined in \cite{rice}. It has a state set $Q$ of
size 16, including a killer state $\kappa$, radius 1 and is defined by the
transitions appearing in Figure~\ref{fig:deffire}: precisely, any transition
which is not in the space-time diagram of the figure produces the killer state
$\kappa$. The complete list of transitions is given in \cite{rice}.

\newcommand\st[4]{{\draw[fill=#4] (#2,#3)
    -- ++(1,0) -- ++(0,1)--++(-1,0)--cycle;\draw (#2,#3)+(.5,.5) node {\Tiny #1};}}
\newcommand\pointB[2]{}
\newcommand\pointF[2]{\st{\#}{#1}{#2}{gray}}
\newcommand\pointS[2]{\st{\#'}{#1}{#2}{gray!50!white}}
\newcommand\pointSp[2]{\st{$\gamma$}{#1}{#2}{yellow!50!white}}
\newcommand\pointXXXXX[2]{\st{$L_1$}{#1}{#2}{blue!50!white}}
\newcommand\pointXXXXXX[2]{\st{$l_1$}{#1}{#2}{blue!50!white}}
\newcommand\pointXXXXXXX[2]{\st{$R_1$}{#1}{#2}{red!50!white}}
\newcommand\pointXXXXXXXX[2]{\st{$r_1$}{#1}{#2}{red!50!white}}
\newcommand\pointXXXXXXXXX[2]{\st{$l_2$}{#1}{#2}{blue!20!white}}
\newcommand\pointXXXXXXXXXX[2]{\st{$l_2$}{#1}{#2}{blue!70!white}}
\newcommand\pointXXXXXXXXXXX[2]{\st{$r_2$}{#1}{#2}{red!20!white}}
\newcommand\pointXXXXXXXXXXXX[2]{\st{$r_2$}{#1}{#2}{red!70!white}}
\newcommand\pointXXXXXXXXXXXXX[2]{\st{X}{#1}{#2}{green!50!white}}
\newcommand\pointXXXXXXXXXXXXXX[2]{\st{Y}{#1}{#2}{green!50!white}}
\newcommand\pointXXXXXXXXXXXXXXX[2]{\st{Z}{#1}{#2}{green!50!white}}

\begin{figure}
  \centering
  \begin{tikzpicture}[scale=.3]
    \pointB{0}{0}\pointB{0}{1}\pointB{0}{2}\pointB{0}{3}\pointB{0}{4}\pointB{0}{5}\pointB{0}{6}\pointXXXXX{0}{7}\pointB{0}{8}\pointXXXXXXXX{0}{9}\pointB{0}{10}\pointB{0}{11}\pointB{0}{12}\pointB{0}{13}\pointXXXXXXXXX{0}{14}\pointXXXXXXXXXX{0}{15}\pointB{0}{16}\pointXXXXXXX{0}{17}\pointXXXXXXXXXXX{0}{18}\pointXXXXXXXXXXXX{0}{19}\pointB{0}{20}\pointB{0}{21}\pointB{0}{22}\pointXXXXXX{0}{23}\pointB{0}{24}\pointXXXXXXX{0}{25}\pointXXXXXXXXXXX{0}{26}\pointXXXXXXXXXXXXXXX{0}{27}\pointB{0}{28}\pointS{0}{29}\pointSp{0}{30}
\pointB{1}{0}\pointB{1}{1}\pointB{1}{2}\pointB{1}{3}\pointB{1}{4}\pointB{1}{5}\pointXXXXX{1}{6}\pointB{1}{7}\pointB{1}{8}\pointB{1}{9}\pointXXXXXXXX{1}{10}\pointB{1}{11}\pointXXXXXXXXX{1}{12}\pointXXXXXXXXXX{1}{13}\pointB{1}{14}\pointB{1}{15}\pointB{1}{16}\pointB{1}{17}\pointXXXXXXX{1}{18}\pointB{1}{19}\pointXXXXXXXXXXX{1}{20}\pointXXXXXXXXXXXX{1}{21}\pointXXXXXX{1}{22}\pointB{1}{23}\pointB{1}{24}\pointB{1}{25}\pointXXXXXXXXXXXXX{1}{26}\pointB{1}{27}\pointF{1}{28}\pointS{1}{29}\pointSp{1}{30}
\pointB{2}{0}\pointB{2}{1}\pointB{2}{2}\pointB{2}{3}\pointB{2}{4}\pointXXXXX{2}{5}\pointB{2}{6}\pointB{2}{7}\pointB{2}{8}\pointB{2}{9}\pointXXXXXXXXX{2}{10}\pointXXXXXXXXXXXXXX{2}{11}\pointB{2}{12}\pointB{2}{13}\pointB{2}{14}\pointB{2}{15}\pointB{2}{16}\pointB{2}{17}\pointB{2}{18}\pointXXXXXXX{2}{19}\pointB{2}{20}\pointXXXXXX{2}{21}\pointXXXXXXXXXXX{2}{22}\pointXXXXXXXXXXXX{2}{23}\pointB{2}{24}\pointXXXXX{2}{25}\pointXXXXXXXXX{2}{26}\pointXXXXXXXXXXXXXX{2}{27}\pointB{2}{28}\pointS{2}{29}\pointSp{2}{30}
\pointB{3}{0}\pointB{3}{1}\pointB{3}{2}\pointB{3}{3}\pointXXXXX{3}{4}\pointB{3}{5}\pointB{3}{6}\pointB{3}{7}\pointXXXXXXXXX{3}{8}\pointXXXXXXXXXX{3}{9}\pointB{3}{10}\pointB{3}{11}\pointXXXXXXXX{3}{12}\pointB{3}{13}\pointB{3}{14}\pointB{3}{15}\pointB{3}{16}\pointB{3}{17}\pointB{3}{18}\pointB{3}{19}\pointXXXXXXXXXXXXX{3}{20}\pointB{3}{21}\pointB{3}{22}\pointB{3}{23}\pointF{3}{24}\pointS{3}{25}\pointS{3}{26}\pointS{3}{27}\pointF{3}{28}\pointS{3}{29}\pointSp{3}{30}
\pointB{4}{0}\pointB{4}{1}\pointB{4}{2}\pointXXXXX{4}{3}\pointB{4}{4}\pointB{4}{5}\pointXXXXXXXXX{4}{6}\pointXXXXXXXXXX{4}{7}\pointB{4}{8}\pointB{4}{9}\pointB{4}{10}\pointB{4}{11}\pointB{4}{12}\pointXXXXXXXX{4}{13}\pointB{4}{14}\pointB{4}{15}\pointB{4}{16}\pointB{4}{17}\pointB{4}{18}\pointXXXXX{4}{19}\pointB{4}{20}\pointXXXXXXXX{4}{21}\pointXXXXXXXXX{4}{22}\pointXXXXXXXXXX{4}{23}\pointB{4}{24}\pointXXXXXXX{4}{25}\pointXXXXXXXXXXX{4}{26}\pointXXXXXXXXXXXXXXX{4}{27}\pointB{4}{28}\pointS{4}{29}\pointSp{4}{30}
\pointB{5}{0}\pointB{5}{1}\pointXXXXX{5}{2}\pointB{5}{3}\pointXXXXXXXXX{5}{4}\pointXXXXXXXXXX{5}{5}\pointB{5}{6}\pointB{5}{7}\pointB{5}{8}\pointB{5}{9}\pointB{5}{10}\pointB{5}{11}\pointB{5}{12}\pointB{5}{13}\pointXXXXXXXX{5}{14}\pointB{5}{15}\pointB{5}{16}\pointB{5}{17}\pointXXXXX{5}{18}\pointB{5}{19}\pointXXXXXXXXX{5}{20}\pointXXXXXXXXXX{5}{21}\pointXXXXXXXX{5}{22}\pointB{5}{23}\pointB{5}{24}\pointB{5}{25}\pointXXXXXXXXXXXXX{5}{26}\pointB{5}{27}\pointF{5}{28}\pointS{5}{29}\pointSp{5}{30}
\pointB{6}{0}\pointXXXXX{6}{1}\pointXXXXXXXXX{6}{2}\pointXXXXXXXXXX{6}{3}\pointB{6}{4}\pointB{6}{5}\pointB{6}{6}\pointB{6}{7}\pointB{6}{8}\pointB{6}{9}\pointB{6}{10}\pointB{6}{11}\pointB{6}{12}\pointB{6}{13}\pointB{6}{14}\pointXXXXXXXX{6}{15}\pointB{6}{16}\pointXXXXX{6}{17}\pointXXXXXXXXX{6}{18}\pointXXXXXXXXXX{6}{19}\pointB{6}{20}\pointB{6}{21}\pointB{6}{22}\pointXXXXXXXX{6}{23}\pointB{6}{24}\pointXXXXX{6}{25}\pointXXXXXXXXX{6}{26}\pointXXXXXXXXXXXXXX{6}{27}\pointB{6}{28}\pointS{6}{29}\pointSp{6}{30}
\pointF{7}{0}\pointS{7}{1}\pointS{7}{2}\pointS{7}{3}\pointS{7}{4}\pointS{7}{5}\pointS{7}{6}\pointS{7}{7}\pointS{7}{8}\pointS{7}{9}\pointS{7}{10}\pointS{7}{11}\pointS{7}{12}\pointS{7}{13}\pointS{7}{14}\pointS{7}{15}\pointF{7}{16}\pointS{7}{17}\pointS{7}{18}\pointS{7}{19}\pointS{7}{20}\pointS{7}{21}\pointS{7}{22}\pointS{7}{23}\pointF{7}{24}\pointS{7}{25}\pointS{7}{26}\pointS{7}{27}\pointF{7}{28}\pointS{7}{29}\pointSp{7}{30}
\pointB{8}{0}\pointXXXXXXX{8}{1}\pointXXXXXXXXXXX{8}{2}\pointXXXXXXXXXXXX{8}{3}\pointB{8}{4}\pointB{8}{5}\pointB{8}{6}\pointB{8}{7}\pointB{8}{8}\pointB{8}{9}\pointB{8}{10}\pointB{8}{11}\pointB{8}{12}\pointB{8}{13}\pointB{8}{14}\pointXXXXXX{8}{15}\pointB{8}{16}\pointXXXXXXX{8}{17}\pointXXXXXXXXXXX{8}{18}\pointXXXXXXXXXXXX{8}{19}\pointB{8}{20}\pointB{8}{21}\pointB{8}{22}\pointXXXXXX{8}{23}\pointB{8}{24}\pointXXXXXXX{8}{25}\pointXXXXXXXXXXX{8}{26}\pointXXXXXXXXXXXXXXX{8}{27}\pointB{8}{28}\pointS{8}{29}\pointSp{8}{30}
\pointB{9}{0}\pointB{9}{1}\pointXXXXXXX{9}{2}\pointB{9}{3}\pointXXXXXXXXXXX{9}{4}\pointXXXXXXXXXXXX{9}{5}\pointB{9}{6}\pointB{9}{7}\pointB{9}{8}\pointB{9}{9}\pointB{9}{10}\pointB{9}{11}\pointB{9}{12}\pointB{9}{13}\pointXXXXXX{9}{14}\pointB{9}{15}\pointB{9}{16}\pointB{9}{17}\pointXXXXXXX{9}{18}\pointB{9}{19}\pointXXXXXXXXXXX{9}{20}\pointXXXXXXXXXXXX{9}{21}\pointXXXXXX{9}{22}\pointB{9}{23}\pointB{9}{24}\pointB{9}{25}\pointXXXXXXXXXXXXX{9}{26}\pointB{9}{27}\pointF{9}{28}\pointS{9}{29}\pointSp{9}{30}
\pointB{10}{0}\pointB{10}{1}\pointB{10}{2}\pointXXXXXXX{10}{3}\pointB{10}{4}\pointB{10}{5}\pointXXXXXXXXXXX{10}{6}\pointXXXXXXXXXXXX{10}{7}\pointB{10}{8}\pointB{10}{9}\pointB{10}{10}\pointB{10}{11}\pointB{10}{12}\pointXXXXXX{10}{13}\pointB{10}{14}\pointB{10}{15}\pointB{10}{16}\pointB{10}{17}\pointB{10}{18}\pointXXXXXXX{10}{19}\pointB{10}{20}\pointXXXXXX{10}{21}\pointXXXXXXXXXXX{10}{22}\pointXXXXXXXXXXXX{10}{23}\pointB{10}{24}\pointXXXXX{10}{25}\pointXXXXXXXXX{10}{26}\pointXXXXXXXXXXXXXX{10}{27}\pointB{10}{28}\pointS{10}{29}\pointSp{10}{30}
\pointB{11}{0}\pointB{11}{1}\pointB{11}{2}\pointB{11}{3}\pointXXXXXXX{11}{4}\pointB{11}{5}\pointB{11}{6}\pointB{11}{7}\pointXXXXXXXXXXX{11}{8}\pointXXXXXXXXXXXX{11}{9}\pointB{11}{10}\pointB{11}{11}\pointXXXXXX{11}{12}\pointB{11}{13}\pointB{11}{14}\pointB{11}{15}\pointB{11}{16}\pointB{11}{17}\pointB{11}{18}\pointB{11}{19}\pointXXXXXXXXXXXXX{11}{20}\pointB{11}{21}\pointB{11}{22}\pointB{11}{23}\pointF{11}{24}\pointS{11}{25}\pointS{11}{26}\pointS{11}{27}\pointF{11}{28}\pointS{11}{29}\pointSp{11}{30}
\pointB{12}{0}\pointB{12}{1}\pointB{12}{2}\pointB{12}{3}\pointB{12}{4}\pointXXXXXXX{12}{5}\pointB{12}{6}\pointB{12}{7}\pointB{12}{8}\pointB{12}{9}\pointXXXXXXXXXXX{12}{10}\pointXXXXXXXXXXXXXXX{12}{11}\pointB{12}{12}\pointB{12}{13}\pointB{12}{14}\pointB{12}{15}\pointB{12}{16}\pointB{12}{17}\pointB{12}{18}\pointXXXXX{12}{19}\pointB{12}{20}\pointXXXXXXXX{12}{21}\pointXXXXXXXXX{12}{22}\pointXXXXXXXXXX{12}{23}\pointB{12}{24}\pointXXXXXXX{12}{25}\pointXXXXXXXXXXX{12}{26}\pointXXXXXXXXXXXXXXX{12}{27}\pointB{12}{28}\pointS{12}{29}\pointSp{12}{30}
\pointB{13}{0}\pointB{13}{1}\pointB{13}{2}\pointB{13}{3}\pointB{13}{4}\pointB{13}{5}\pointXXXXXXX{13}{6}\pointB{13}{7}\pointB{13}{8}\pointB{13}{9}\pointXXXXXX{13}{10}\pointB{13}{11}\pointXXXXXXXXXXX{13}{12}\pointXXXXXXXXXXXX{13}{13}\pointB{13}{14}\pointB{13}{15}\pointB{13}{16}\pointB{13}{17}\pointXXXXX{13}{18}\pointB{13}{19}\pointXXXXXXXXX{13}{20}\pointXXXXXXXXXX{13}{21}\pointXXXXXXXX{13}{22}\pointB{13}{23}\pointB{13}{24}\pointB{13}{25}\pointXXXXXXXXXXXXX{13}{26}\pointB{13}{27}\pointF{13}{28}\pointS{13}{29}\pointSp{13}{30}
\pointB{14}{0}\pointB{14}{1}\pointB{14}{2}\pointB{14}{3}\pointB{14}{4}\pointB{14}{5}\pointB{14}{6}\pointXXXXXXX{14}{7}\pointB{14}{8}\pointXXXXXX{14}{9}\pointB{14}{10}\pointB{14}{11}\pointB{14}{12}\pointB{14}{13}\pointXXXXXXXXXXX{14}{14}\pointXXXXXXXXXXXX{14}{15}\pointB{14}{16}\pointXXXXX{14}{17}\pointXXXXXXXXX{14}{18}\pointXXXXXXXXXX{14}{19}\pointB{14}{20}\pointB{14}{21}\pointB{14}{22}\pointXXXXXXXX{14}{23}\pointB{14}{24}\pointXXXXX{14}{25}\pointXXXXXXXXX{14}{26}\pointXXXXXXXXXXXXXX{14}{27}\pointB{14}{28}\pointS{14}{29}\pointSp{14}{30}
\pointB{15}{0}\pointB{15}{1}\pointB{15}{2}\pointB{15}{3}\pointB{15}{4}\pointB{15}{5}\pointB{15}{6}\pointB{15}{7}\pointXXXXXXXXXXXXX{15}{8}\pointB{15}{9}\pointB{15}{10}\pointB{15}{11}\pointB{15}{12}\pointB{15}{13}\pointB{15}{14}\pointB{15}{15}\pointF{15}{16}\pointS{15}{17}\pointS{15}{18}\pointS{15}{19}\pointS{15}{20}\pointS{15}{21}\pointS{15}{22}\pointS{15}{23}\pointF{15}{24}\pointS{15}{25}\pointS{15}{26}\pointS{15}{27}\pointF{15}{28}\pointS{15}{29}\pointSp{15}{30}
\pointB{16}{0}\pointB{16}{1}\pointB{16}{2}\pointB{16}{3}\pointB{16}{4}\pointB{16}{5}\pointB{16}{6}\pointXXXXX{16}{7}\pointB{16}{8}\pointXXXXXXXX{16}{9}\pointB{16}{10}\pointB{16}{11}\pointB{16}{12}\pointB{16}{13}\pointXXXXXXXXX{16}{14}\pointXXXXXXXXXX{16}{15}\pointB{16}{16}\pointXXXXXXX{16}{17}\pointXXXXXXXXXXX{16}{18}\pointXXXXXXXXXXXX{16}{19}\pointB{16}{20}\pointB{16}{21}\pointB{16}{22}\pointXXXXXX{16}{23}\pointB{16}{24}\pointXXXXXXX{16}{25}\pointXXXXXXXXXXX{16}{26}\pointXXXXXXXXXXXXXXX{16}{27}\pointB{16}{28}\pointS{16}{29}\pointSp{16}{30}
\pointB{17}{0}\pointB{17}{1}\pointB{17}{2}\pointB{17}{3}\pointB{17}{4}\pointB{17}{5}\pointXXXXX{17}{6}\pointB{17}{7}\pointB{17}{8}\pointB{17}{9}\pointXXXXXXXX{17}{10}\pointB{17}{11}\pointXXXXXXXXX{17}{12}\pointXXXXXXXXXX{17}{13}\pointB{17}{14}\pointB{17}{15}\pointB{17}{16}\pointB{17}{17}\pointXXXXXXX{17}{18}\pointB{17}{19}\pointXXXXXXXXXXX{17}{20}\pointXXXXXXXXXXXX{17}{21}\pointXXXXXX{17}{22}\pointB{17}{23}\pointB{17}{24}\pointB{17}{25}\pointXXXXXXXXXXXXX{17}{26}\pointB{17}{27}\pointF{17}{28}\pointS{17}{29}\pointSp{17}{30}
\pointB{18}{0}\pointB{18}{1}\pointB{18}{2}\pointB{18}{3}\pointB{18}{4}\pointXXXXX{18}{5}\pointB{18}{6}\pointB{18}{7}\pointB{18}{8}\pointB{18}{9}\pointXXXXXXXXX{18}{10}\pointXXXXXXXXXXXXXX{18}{11}\pointB{18}{12}\pointB{18}{13}\pointB{18}{14}\pointB{18}{15}\pointB{18}{16}\pointB{18}{17}\pointB{18}{18}\pointXXXXXXX{18}{19}\pointB{18}{20}\pointXXXXXX{18}{21}\pointXXXXXXXXXXX{18}{22}\pointXXXXXXXXXXXX{18}{23}\pointB{18}{24}\pointXXXXX{18}{25}\pointXXXXXXXXX{18}{26}\pointXXXXXXXXXXXXXX{18}{27}\pointB{18}{28}\pointS{18}{29}\pointSp{18}{30}
\pointB{19}{0}\pointB{19}{1}\pointB{19}{2}\pointB{19}{3}\pointXXXXX{19}{4}\pointB{19}{5}\pointB{19}{6}\pointB{19}{7}\pointXXXXXXXXX{19}{8}\pointXXXXXXXXXX{19}{9}\pointB{19}{10}\pointB{19}{11}\pointXXXXXXXX{19}{12}\pointB{19}{13}\pointB{19}{14}\pointB{19}{15}\pointB{19}{16}\pointB{19}{17}\pointB{19}{18}\pointB{19}{19}\pointXXXXXXXXXXXXX{19}{20}\pointB{19}{21}\pointB{19}{22}\pointB{19}{23}\pointF{19}{24}\pointS{19}{25}\pointS{19}{26}\pointS{19}{27}\pointF{19}{28}\pointS{19}{29}\pointSp{19}{30}
\pointB{20}{0}\pointB{20}{1}\pointB{20}{2}\pointXXXXX{20}{3}\pointB{20}{4}\pointB{20}{5}\pointXXXXXXXXX{20}{6}\pointXXXXXXXXXX{20}{7}\pointB{20}{8}\pointB{20}{9}\pointB{20}{10}\pointB{20}{11}\pointB{20}{12}\pointXXXXXXXX{20}{13}\pointB{20}{14}\pointB{20}{15}\pointB{20}{16}\pointB{20}{17}\pointB{20}{18}\pointXXXXX{20}{19}\pointB{20}{20}\pointXXXXXXXX{20}{21}\pointXXXXXXXXX{20}{22}\pointXXXXXXXXXX{20}{23}\pointB{20}{24}\pointXXXXXXX{20}{25}\pointXXXXXXXXXXX{20}{26}\pointXXXXXXXXXXXXXXX{20}{27}\pointB{20}{28}\pointS{20}{29}\pointSp{20}{30}
\pointB{21}{0}\pointB{21}{1}\pointXXXXX{21}{2}\pointB{21}{3}\pointXXXXXXXXX{21}{4}\pointXXXXXXXXXX{21}{5}\pointB{21}{6}\pointB{21}{7}\pointB{21}{8}\pointB{21}{9}\pointB{21}{10}\pointB{21}{11}\pointB{21}{12}\pointB{21}{13}\pointXXXXXXXX{21}{14}\pointB{21}{15}\pointB{21}{16}\pointB{21}{17}\pointXXXXX{21}{18}\pointB{21}{19}\pointXXXXXXXXX{21}{20}\pointXXXXXXXXXX{21}{21}\pointXXXXXXXX{21}{22}\pointB{21}{23}\pointB{21}{24}\pointB{21}{25}\pointXXXXXXXXXXXXX{21}{26}\pointB{21}{27}\pointF{21}{28}\pointS{21}{29}\pointSp{21}{30}
\pointB{22}{0}\pointXXXXX{22}{1}\pointXXXXXXXXX{22}{2}\pointXXXXXXXXXX{22}{3}\pointB{22}{4}\pointB{22}{5}\pointB{22}{6}\pointB{22}{7}\pointB{22}{8}\pointB{22}{9}\pointB{22}{10}\pointB{22}{11}\pointB{22}{12}\pointB{22}{13}\pointB{22}{14}\pointXXXXXXXX{22}{15}\pointB{22}{16}\pointXXXXX{22}{17}\pointXXXXXXXXX{22}{18}\pointXXXXXXXXXX{22}{19}\pointB{22}{20}\pointB{22}{21}\pointB{22}{22}\pointXXXXXXXX{22}{23}\pointB{22}{24}\pointXXXXX{22}{25}\pointXXXXXXXXX{22}{26}\pointXXXXXXXXXXXXXX{22}{27}\pointB{22}{28}\pointS{22}{29}\pointSp{22}{30}
\pointF{23}{0}\pointS{23}{1}\pointS{23}{2}\pointS{23}{3}\pointS{23}{4}\pointS{23}{5}\pointS{23}{6}\pointS{23}{7}\pointS{23}{8}\pointS{23}{9}\pointS{23}{10}\pointS{23}{11}\pointS{23}{12}\pointS{23}{13}\pointS{23}{14}\pointS{23}{15}\pointF{23}{16}\pointS{23}{17}\pointS{23}{18}\pointS{23}{19}\pointS{23}{20}\pointS{23}{21}\pointS{23}{22}\pointS{23}{23}\pointF{23}{24}\pointS{23}{25}\pointS{23}{26}\pointS{23}{27}\pointF{23}{28}\pointS{23}{29}\pointSp{23}{30}
\pointB{24}{0}\pointXXXXXXX{24}{1}\pointXXXXXXXXXXX{24}{2}\pointXXXXXXXXXXXX{24}{3}\pointB{24}{4}\pointB{24}{5}\pointB{24}{6}\pointB{24}{7}\pointB{24}{8}\pointB{24}{9}\pointB{24}{10}\pointB{24}{11}\pointB{24}{12}\pointB{24}{13}\pointB{24}{14}\pointXXXXXX{24}{15}\pointB{24}{16}\pointXXXXXXX{24}{17}\pointXXXXXXXXXXX{24}{18}\pointXXXXXXXXXXXX{24}{19}\pointB{24}{20}\pointB{24}{21}\pointB{24}{22}\pointXXXXXX{24}{23}\pointB{24}{24}\pointXXXXXXX{24}{25}\pointXXXXXXXXXXX{24}{26}\pointXXXXXXXXXXXXXXX{24}{27}\pointB{24}{28}\pointS{24}{29}\pointSp{24}{30}
\pointB{25}{0}\pointB{25}{1}\pointXXXXXXX{25}{2}\pointB{25}{3}\pointXXXXXXXXXXX{25}{4}\pointXXXXXXXXXXXX{25}{5}\pointB{25}{6}\pointB{25}{7}\pointB{25}{8}\pointB{25}{9}\pointB{25}{10}\pointB{25}{11}\pointB{25}{12}\pointB{25}{13}\pointXXXXXX{25}{14}\pointB{25}{15}\pointB{25}{16}\pointB{25}{17}\pointXXXXXXX{25}{18}\pointB{25}{19}\pointXXXXXXXXXXX{25}{20}\pointXXXXXXXXXXXX{25}{21}\pointXXXXXX{25}{22}\pointB{25}{23}\pointB{25}{24}\pointB{25}{25}\pointXXXXXXXXXXXXX{25}{26}\pointB{25}{27}\pointF{25}{28}\pointS{25}{29}\pointSp{25}{30}
\pointB{26}{0}\pointB{26}{1}\pointB{26}{2}\pointXXXXXXX{26}{3}\pointB{26}{4}\pointB{26}{5}\pointXXXXXXXXXXX{26}{6}\pointXXXXXXXXXXXX{26}{7}\pointB{26}{8}\pointB{26}{9}\pointB{26}{10}\pointB{26}{11}\pointB{26}{12}\pointXXXXXX{26}{13}\pointB{26}{14}\pointB{26}{15}\pointB{26}{16}\pointB{26}{17}\pointB{26}{18}\pointXXXXXXX{26}{19}\pointB{26}{20}\pointXXXXXX{26}{21}\pointXXXXXXXXXXX{26}{22}\pointXXXXXXXXXXXX{26}{23}\pointB{26}{24}\pointXXXXX{26}{25}\pointXXXXXXXXX{26}{26}\pointXXXXXXXXXXXXXX{26}{27}\pointB{26}{28}\pointS{26}{29}\pointSp{26}{30}
\pointB{27}{0}\pointB{27}{1}\pointB{27}{2}\pointB{27}{3}\pointXXXXXXX{27}{4}\pointB{27}{5}\pointB{27}{6}\pointB{27}{7}\pointXXXXXXXXXXX{27}{8}\pointXXXXXXXXXXXX{27}{9}\pointB{27}{10}\pointB{27}{11}\pointXXXXXX{27}{12}\pointB{27}{13}\pointB{27}{14}\pointB{27}{15}\pointB{27}{16}\pointB{27}{17}\pointB{27}{18}\pointB{27}{19}\pointXXXXXXXXXXXXX{27}{20}\pointB{27}{21}\pointB{27}{22}\pointB{27}{23}\pointF{27}{24}\pointS{27}{25}\pointS{27}{26}\pointS{27}{27}\pointF{27}{28}\pointS{27}{29}\pointSp{27}{30}
\pointB{28}{0}\pointB{28}{1}\pointB{28}{2}\pointB{28}{3}\pointB{28}{4}\pointXXXXXXX{28}{5}\pointB{28}{6}\pointB{28}{7}\pointB{28}{8}\pointB{28}{9}\pointXXXXXXXXXXX{28}{10}\pointXXXXXXXXXXXXXXX{28}{11}\pointB{28}{12}\pointB{28}{13}\pointB{28}{14}\pointB{28}{15}\pointB{28}{16}\pointB{28}{17}\pointB{28}{18}\pointXXXXX{28}{19}\pointB{28}{20}\pointXXXXXXXX{28}{21}\pointXXXXXXXXX{28}{22}\pointXXXXXXXXXX{28}{23}\pointB{28}{24}\pointXXXXXXX{28}{25}\pointXXXXXXXXXXX{28}{26}\pointXXXXXXXXXXXXXXX{28}{27}\pointB{28}{28}\pointS{28}{29}\pointSp{28}{30}
\pointB{29}{0}\pointB{29}{1}\pointB{29}{2}\pointB{29}{3}\pointB{29}{4}\pointB{29}{5}\pointXXXXXXX{29}{6}\pointB{29}{7}\pointB{29}{8}\pointB{29}{9}\pointXXXXXX{29}{10}\pointB{29}{11}\pointXXXXXXXXXXX{29}{12}\pointXXXXXXXXXXXX{29}{13}\pointB{29}{14}\pointB{29}{15}\pointB{29}{16}\pointB{29}{17}\pointXXXXX{29}{18}\pointB{29}{19}\pointXXXXXXXXX{29}{20}\pointXXXXXXXXXX{29}{21}\pointXXXXXXXX{29}{22}\pointB{29}{23}\pointB{29}{24}\pointB{29}{25}\pointXXXXXXXXXXXXX{29}{26}\pointB{29}{27}\pointF{29}{28}\pointS{29}{29}\pointSp{29}{30}
\pointB{30}{0}\pointB{30}{1}\pointB{30}{2}\pointB{30}{3}\pointB{30}{4}\pointB{30}{5}\pointB{30}{6}\pointXXXXXXX{30}{7}\pointB{30}{8}\pointXXXXXX{30}{9}\pointB{30}{10}\pointB{30}{11}\pointB{30}{12}\pointB{30}{13}\pointXXXXXXXXXXX{30}{14}\pointXXXXXXXXXXXX{30}{15}\pointB{30}{16}\pointXXXXX{30}{17}\pointXXXXXXXXX{30}{18}\pointXXXXXXXXXX{30}{19}\pointB{30}{20}\pointB{30}{21}\pointB{30}{22}\pointXXXXXXXX{30}{23}\pointB{30}{24}\pointXXXXX{30}{25}\pointXXXXXXXXX{30}{26}\pointXXXXXXXXXXXXXX{30}{27}\pointB{30}{28}\pointS{30}{29}\pointSp{30}{30}
\pointB{31}{0}\pointB{31}{1}\pointB{31}{2}\pointB{31}{3}\pointB{31}{4}\pointB{31}{5}\pointB{31}{6}\pointB{31}{7}\pointXXXXXXXXXXXXX{31}{8}\pointB{31}{9}\pointB{31}{10}\pointB{31}{11}\pointB{31}{12}\pointB{31}{13}\pointB{31}{14}\pointB{31}{15}\pointF{31}{16}\pointS{31}{17}\pointS{31}{18}\pointS{31}{19}\pointS{31}{20}\pointS{31}{21}\pointS{31}{22}\pointS{31}{23}\pointF{31}{24}\pointS{31}{25}\pointS{31}{26}\pointS{31}{27}\pointF{31}{28}\pointS{31}{29}\pointSp{31}{30}
  \end{tikzpicture}
  \caption{The 16-state firing squad of \cite{rice}. Empty spaces represent the
    blank state.}
  \label{fig:deffire}
\end{figure}
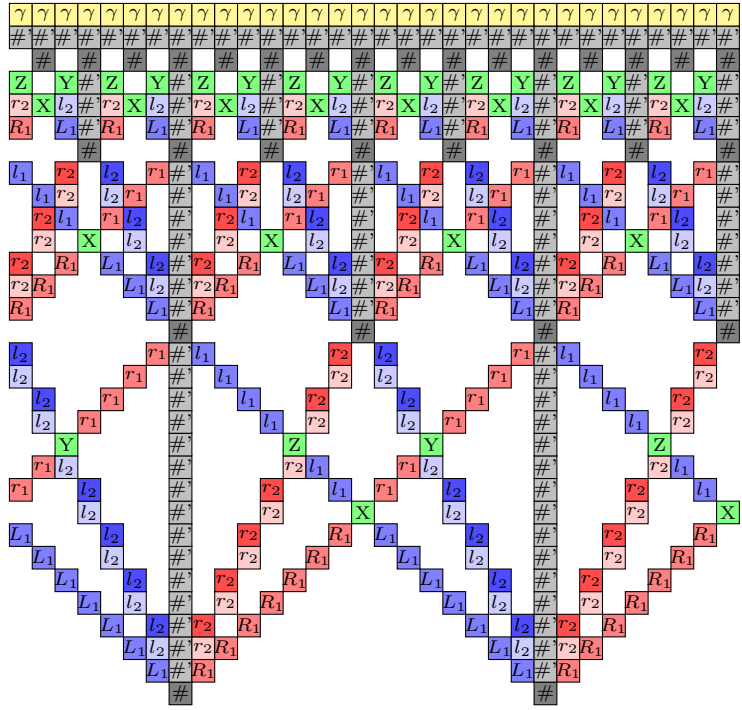

We are interested in history diagrams in $X_S$, \textit{i.e.} mapping from
$\Z\times\N$ to $Q$ of the form: ${(z,t)\mapsto x^t(z)}$ where $(x^t)$ is a sequence of configuration in $X_S$
such that $S(x^{t+1})=x^{t}$.  We call them \textit{valid} history diagrams.

When restricted to $X_S$, the behavior of $S$ is easier to understand
via a signal/collision evolving in a quiescent background. More
precisely, the background is uniform and made of blank states (denoted
$B$ in the sequel) and the signals involved are:
\begin{center}
  \begin{tabular}{r|c|c|c|c|c|c|c}
    \textit{signal} &$L_1$ & $l_1$ & $l_2$ & \#' & $r_2$ & $r_1$ & $R_1$\\
    \hline
    \textit{speed} & -1 & -1 & -1/2 & 0 & 1/2 & 1 & 1
  \end{tabular}
\end{center}

The valid collisions are:
\begin{center}
  \begin{minipage}{.48\linewidth}
    \begin{align*}
      L_1+R_1 &\rightarrow l_1 + r_1\\
      l_1 + r_2 &\rightarrow l_1 + r_2\\
      r_1 + l_2 &\rightarrow r_1 + l_2
    \end{align*}
  \end{minipage}
  \begin{minipage}{.48\linewidth}
    \begin{align*}
      l_2+r_2 &\rightarrow \#\\
      l_1+\#'+r_1 &\rightarrow \#\\
      \# &\rightarrow L_1+l_2+\#'+r_2+R_1
    \end{align*}
  \end{minipage}
\end{center}
Any other intersection of signal is invalid (it raises a $\kappa$
state). Moreover, the last collision rule (starting from a single $\#$) is valid
only if the $\#$ is distant from any other $\#$ by at least $3$ cells: if they
are $1$ cell away, $3$ adjacent $\#'$ are generated; if they are $2$ cells away, a
$\kappa$ is generated $2$ steps later.

To simplify proofs, we will often make reasonings over (portions of)
``Euclidean'' versions of history diagrams. A Euclidean history diagram
is a set of labelled points and labelled (half-)lines or segments in
$\mathbb{R}^2$ satisfying the following rules:
\begin{itemize}
\item points are only at integer coordinates ($\Z^2$) and labelled by $\#$;
\item (half-)lines and segments correspond to signals listed above (label and slope correspond);
\item any intersection between lines or segments follow the collision rules above.
\end{itemize}

\begin{lem}
  \label{euclidean}
  To each history diagram $D$, we can associate a valid Euclidean history
  diagram $E$ such that, at any integer coordinate of $E$ containing a point
  ($\#$) or a single signal, the label gives the state of the corresponding
  position in $D$.
\end{lem}

This lemma allows the following proof scheme (used several times below):
supposing by sake of contradiction that some word $w$ occurs in a history
diagram, we make a reasoning on the corresponding Euclidean diagram, we get a
contradiction and finally deduce that no history diagram exists which contain
the word $w$, and therefore that $w$ is not in the language of $X_S$.

$S$ satisfies points (1) and (2) of Lemma~\ref{l:fsquad} as shown in
\cite[Prop. 4.3]{rice}. We give below a complete characterization of
the language of $X_S$ which shows that it is NL-recognizable.

\begin{lemma}
\label{columns-lemma}
Consider a history diagram containing a word $w\in\#Q^*\#'$ at time $t_0$. Let
$z_1$ (resp. $z_2$) be the cell where the first (resp. the last) letter of $w$
occurs. We suppose in addition that the left $\#$ of $w$ was created by a $l_1$
signal. Let $t_1$ be the first time step in the past when the cell $z_1$ is in
state $\#$, and $t_2$ be the first time step in the past when the cell $z_2$ is
in state $\#$. Then, both $t_1$ and $t_2$ exist.

\end{lemma}

We denote by $\Sigma$ the set ${Q'\setminus\{\#,\#'\}}$.

\begin{lemma}
  \label{l1r1-prop}
  There is no history diagram containing a word $w$ of the form $\#\Sigma^*\#'$,
  where the left $\#$ is created by $r_1/l_1$ signals.
\end{lemma}

\begin{lemma}
  \label{l2r2-prop}
  There is no history diagram containing a word $w$ of the form $\#\Sigma^*\#'$,
  where the left $\#$ was created by a $l_2/r_2$ pair of signals.
\end{lemma}

\begin{lemma}
  \label{periodic-prop}
  Any configuration from $X_S$ with at least two $\#$ is of the following form,
  for some value of $n$: ${^\omega(\#B^n)^\omega}$
\end{lemma}

\begin{lemma}
  \label{validlanguage}
  Let $L$ be the language of configurations from $X_S$ admitting an
  history diagram where two $\#$ occur at some time $t$ in the
  past. Then $L\in\textsc{nl}$ -- recognizable in logarithmic space.
\end{lemma}

\begin{lemma}
  \label{one-ptime}
  The language of configurations from $X_S$ which contain only one state in
  $\{\#,\#'\}$ is also in $\textsc{nl}$.
\end{lemma}

\begin{lemma}
  \label{neutral-sharpsprime}
  Let $L$ be the language of configurations from $X_S$ with two or
  more $\#'$ and having a history diagram with no $\#$.  Then $L$ is
  regular.
\end{lemma}

\begin{lemma}
  \label{zero-regular}
  The language of the configurations from $X_S$ without any $\#$ or
  $\#'$ is regular.
\end{lemma}

\begin{proposition}
\label{prop:simplefirelimit}
The language of $X_S$ is in \textsc{nl}.
\begin{proof}
  There are several cases, and the disjunction on configurations allows to express
  the language of $X_S$ as a union of 'simple' \textsc{nl} languages.
  \begin{enumerate}
    \item Configurations with $\#$s or $\#$s. We can descibe the set of
      these configurations by :
      \[^\omega\{L_1,l_1,B\}\{l_2,B\}^*A\{r_2,B\}^*\{R_1,r_1,B\}^\omega\]
      where $A$ is one of the following (possibly infinite) configurations:
      \begin{enumerate}
      \item $A$ has exactly one state in $\{\#,\#'\}$.
        We conclude in this case with Lemma~\ref{one-ptime}.
      \item $A$ has one $\#$, and at least one other $\#$ or
        $\#'$. Lemmas~\ref{l1r1-prop}, \ref{l2r2-prop} and
        \ref{periodic-prop} show that the configuration satisfy the
        hypothesis of Lemma~\ref{validlanguage}, which allows to
        conclude.\label{case:twosharps}
      \item $A$ has at least two $\#'$ but do not contain any signal.
        Then Lemma~\ref{neutral-sharpsprime} conclude.
      \item $A$ has at least two $\#'$, along with some signal(s)
        between two $\#'$. Denote by $c$ the global configuration in
        this case.  We can simply go back a few steps in the past to
        find out a configuration $c'$ of case
        \ref{case:twosharps}. Then we can apply Lemma
       ~\ref{validlanguage} to $c$.
      \end{enumerate}
    \item Configurations without $\#$s nor $\#'$s. This case is treated in
      Lemma~\ref{zero-regular}.
  \end{enumerate}
  Thus, since we have described above why each of the possible languages could be recognized
  in $\textsc{nl}$, we can just build a non-deterministic machine beginning by making a
  non-deterministic choice between all of these machines, then doing the
  computation of the chosen one.
\end{proof}

\end{proposition}

\section*{Conclusion}\label{sec:persp}

We have thus achieved results implying that both limit set and column factors complexities are strongly linked to the factor simulation hierarchy; on the other hand, they are rather orthogonal to the sub-automaton simulation hierarchy.

Many open questions remain.
\begin{itemize}
\item We have obtained that universality was not forbidden by some
  rather strong constraints either on the limit set, or
  ``orthogonally'', on the column factors.  A natural question is
  whether we can constrain both at the same time. The two constructions
  may possibly be composed together, at the price of a (yet) more
  difficult proof of the NL-recognizability of the limit set.
\item Similarly, we believe that our results still hold when alphabets are restricted to $\{0,1\}$ but at the price of a more technical proof.
\item Is there an intrinsically universal CA with an SFT limit set? Following our
  construction, this raises immediately the following question: is there a
  firing-squad CA with an SFT limit set?
\item What kind of limit system can an intrinsically universal CA have? Can it be injective?
\item Is injectivity, expansivity of CA preserved by factor maps?
\item Is it enough, for a CA to be a factor of another CA, that the corresponding column factors with some given width be linked by a factor map?
\item Is there, for some complexity level $\lambda$, an equivalence
  class for the sub-automaton simulation (with space-time rescalings)
  of which all the elements have limit sets of complexity $\lambda$?
\end{itemize}

\bibliographystyle{alpha}
\bibliography{simpluniv}

\begin{thebibliography}{DMOT10b}

\bibitem[BM97]{oexp}
Fran\c{c}ois Blanchard and Alejandro Maass.
\newblock Dynamical properties of expansive one-sided cellular automata.
\newblock {\em Israel Journal of Mathematics}, 99:149--174, 1997.

\bibitem[CFG10]{utrace}
Julien Cervelle, Enrico Formenti, and Pierre Guillon.
\newblock Ultimate traces cellular automata.
\newblock In Jean-Yves Marion, editor, {\em $27^{\text{th}}$ International
  Symposium on Theoretical Aspects of Computer Science (STACS'10)}, Nancy,
  March 2010.

\bibitem[{\v C}PY89]{cpyu}
Karel {\v C}{ulik II}, Jan~K. Pachl, and Sheng Yu.
\newblock On the limit sets of cellular automata.
\newblock {\em SIAM Journal on Computing}, 18(4):831--842, 1989.

\bibitem[DMOT10a]{bulk1}
Marianne Delorme, Jacques Mazoyer, Nicolas Ollinger, and Guillaume Theyssier.
\newblock Bulking {I}: an abstract theory of bulking.
\newblock HAL:hal-00451732, January 2010.

\bibitem[DMOT10b]{bulk2}
Marianne Delorme, Jacques Mazoyer, Nicolas Ollinger, and Guillaume Theyssier.
\newblock Bulking {II}: Classifications of cellular automata.
\newblock HAL:hal-00451729, January 2010.

\bibitem[Fio00]{edensof}
Francesca Fiorenzi.
\newblock The {G}arden of {E}den theorem for sofic shifts.
\newblock {\em Pure Mathematics and Applications}, 11(3):471--484, 2000.

\bibitem[Gil88]{notes}
Robert~H. Gilman.
\newblock Notes on cellular automata.
\newblock manuscript, 1988.

\bibitem[GMM93]{limuniv}
Eric Goles, Alejandro Maass, and Servet Martínez.
\newblock On the limit set of some universal cellular automata.
\newblock {\em Theoretical Computer Science}, 110:53--78, 1993.

\bibitem[Hur90]{langlim2}
Lyman~P. Hurd.
\newblock Nonrecursive cellular automata invariant sets.
\newblock {\em Complex Systems}, 4:131--138, 1990.

\bibitem[Kar94]{rice}
Jarkko Kari.
\newblock {R}ice's theorem for the limit sets of cellular automata.
\newblock {\em Theoretical Computer Science}, 127(2):229--254, 1994.

\bibitem[K{\r u}r97]{classif}
Petr K{\r u}rka.
\newblock Languages, equicontinuity and attractors in cellular automata.
\newblock {\em Ergodic Theory \& Dynamical Systems}, 17:417--433, 1997.

\bibitem[K{\r u}r03]{kurka}
Petr K{\r u}rka.
\newblock {\em Topological and symbolic dynamics}.
\newblock Soci{\'e}t{\'e} Math{\'e}matique de France, 2003.

\bibitem[Nas95]{nasu}
Masakazu Nasu.
\newblock {\em Textile Systems for Endomorphisms and Automorphisms of the
  Shift}, volume 114 of {\em Memoirs of the American Mathematical Society}.
\newblock American Mathematical Society, Providence, Rhode Island, March 1995.

\bibitem[Oll02]{ollinger}
Nicolas Ollinger.
\newblock {\em Automates cellulaires: structures}.
\newblock PhD thesis, École Normale Supérieure de Lyon, December 2002.

\bibitem[Oll03]{univind}
Nicolas Ollinger.
\newblock The intrinsic universality problem of one-dimensional cellular
  automata.
\newblock In Helmut Alt and Michel Habib, editors, {\em $20^\text{th}$ Annual
  Symposium on Theoretical Aspects of Computer Science (STACS'03)}, volume 2607
  of {\em Lecture Notes in Computer Science}, pages 632--641, Berlin, Germany,
  February 2003. Springer-Verlag.

\bibitem[Taa07]{revlim}
Siamak Taati.
\newblock Cellular automata reversible over limit set.
\newblock {\em Journal of Cellular Automata}, 2(2):167--177, 2007.

\bibitem[The05]{theyssier}
Guillaume Theyssier.
\newblock {\em Automates cellulaires: un modèle de complexités}.
\newblock PhD thesis, École Normale Supérieure de Lyon, December 2005.

\bibitem[vN66]{vneumann2}
John von Neumann.
\newblock {\em Theory of Self-Reproducing Automata}.
\newblock University of Illinois Press, Champaign, IL, USA, 1966.

\bibitem[Wei73]{weiss}
Benjamin Weiss.
\newblock Subshifts of finite type and sofic systems.
\newblock {\em Monatshefte f{\"u}r Mathematik}, 77(5):462--474, 1973.

\end{thebibliography}

\newpage
\appendix

\section{Proof of Lemma~\ref{euclidean}}
\begin{proof}
  There is no rounding problem for signals of slope $-1$, $0$ and $1$ (positions
  of cells of discrete signals correspond exactly to integer position on the
  continuous signal). Concerning slopes $1/2$ and $-1/2$, $S$ is such that when
  two $\#$ at position $p_1$ and $p_2$ are connected by a $1/2$ signal in a
  discrete space-time diagram, the displacement vector $p_2-p_1$ is always of
  the form $(n,2n)$ (the situation is similar for slope $-1/2$). More precisely,
  one can check that the set of cells ${\{p_1+(i,2i) : i\in\mathbb{N}, i\leq
    n\}}$ are in a $r_2$ signal state, and it corresponds exactly to the integer
  positions on the continuous segment from $p_1$ to $p_2$.
\end{proof}

\section{Proof of Lemma~\ref{columns-lemma}}
\begin{proof}
  First let us prove that $t_1$ exists. Let us assume that there is no $\#$ in the
  past of cell $z_2$. The only possible past of this cell, in this case, is
  necessarily an infinite column of $\#'$s. But we assumed that the left $\#$ was
  created by an $l_1$ signal, which necesarily either crossed this column, or was
  generated by it. In both cases, this is a contradiction.

  Now that we know this, we prove that $t_2$ exists. Let us assume that cell
  $z_1$ has only $\#'$s in its past. Since the $\#$ in the right column was
  necesarily created by a signal coming from the left, be it $r_1$ or $r_2$;
  this implies that the signal should have ``crossed'' the left column, which
  is impossible. Thus, there is at least one $\#$ in this column.
\end{proof}

\section{Proof of Lemma~\ref{l1r1-prop}}
  \begin{proof}
    Let us call $\mathcal Z$ the $l_1$ signal that created the left $\#$ of $w$ at $t_0$. 
    According to Lemma~\ref{columns-lemma}, there is some time step $t_1$
    in which a $\#$ appears in the past of the right $\#'$ (i.e. of column $\mathcal C$).
    Moreover, let $t_2$ be the most recent step in which a $\#$ appears in the past of
    the $\#$ on column $\mathcal B$ at $t_0$.
    There are two cases:
    \begin{itemize}
    \item This $\#$ in column $\mathcal C$ 
      was created by an $r_1$ signal ${\mathcal Z}_1$. In this case,
      this $\#$ is necessarily the same as the one that generated $\mathcal Z$. Else,
      applying the same reasonning as in Lemma~\ref{columns-lemma}, there would be another
      $\#$ between $\mathcal Z$, and the $\#$ at $t_1$, and then
      ${\mathcal Z}_1$ would have crossed
      both $\mathcal Z$ and the $L_1$ signal emitted by this intermediate $\#$.



      We now have two subcases:
      \begin{itemize}
      \item The $\#$ at time $t_2$ in column $\mathcal B$ was created by
        $l_2/r_2$. We deduce that the $\#$ on column $\mathcal C$ at time $t_1$ was
        created by an $l_1/r_1$ pair of signals: otherwise we would have a
        $l_2/r_2$ intersection between the two columns,
        which is only possible on a $\#$,
        but this would create a column of $\#'$s visible at time $t_0$, which
        contradicts the hypothesis. Thus we can deduce that:
        \begin{enumerate}
        \item Column $\mathcal C$ existed before $t_1$, and we can thus consider the
          first $\#$ on it before $t_1$. This is the $\#$ which emitted the $l_2$ signal
          creating the $\#$ on column $\mathcal B$ at $t_2$, otherwise a $l_1$ or $l_2$
          signal emmitted by this $\#$ would
          intersect column $\mathcal B$ between $t_0$ and $t_2$, which is
          impossible by minimality of $t_2$.
        \item the $r_2$ signal arriving on the $\#$ of column $\mathcal B$ at time
          $t_2$ and the $r_1$ arriving on the $\#$ of column $\mathcal C$ at time
          $t_1$ meet in the past on some $\#$. This allows to infer the existence of
          a third column $\mathcal A$ on the left of the two previous ones;
        \end{enumerate}
        We thus have the following situation:
          
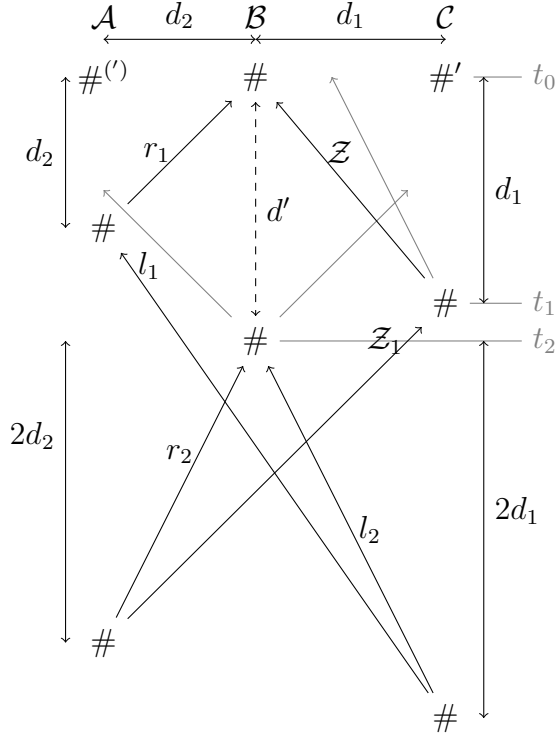
\begin{center}
            \begin{tikzpicture}
              \draw(0,0)node(t0p0)[fill=white]{$\#^{(')}$}
              (2,0)node(t0p1)[fill=white]{$\#$}
              (4.5,0)node(t0p2)[fill=white]{$\#'$};

              \draw(0,-2)node(t3p0)[fill=white]{$\#$};
              \draw(2,-3.5)node(t2p1)[fill=white]{$\#$};
              \draw(4.5,-3)node(t1p2)[fill=white]{$\#$};

              \draw[gray](t0p2)--(5.5,0)node[anchor=west]{$t_0$};
              \draw[gray](t1p2)--(5.5,-3)node[anchor=west]{$t_1$};
              \draw[gray](t2p1)--(5.5,-3.5)node[anchor=west]{$t_2$};
              \draw[->,gray](t1p2)--(3,0);
              \draw[->,gray](t2p1)--(4,-1.5);
              \draw[->,gray](t2p1)--(0,-1.5);

              \draw[->](t3p0)--(t0p1);
              \draw[->](t1p2)--(t0p1);

              \draw(0.7,-1)node{$r_1$};
              \draw(3.1,-1)node{$\mathcal Z$}; 
              
              \draw(0,-7.5)node(t4p0)[fill=white]{$\#$};
              \draw(4.5,-8.5)node(t5p2)[fill=white]{$\#$};
              \draw[->](t4p0)--(t2p1);
              \draw[->](t5p2)--(t2p1);

              \draw[->](t5p2)--(t3p0);
              \draw[->](t4p0)--(t1p2);
              \draw(0.6,-2.5)node{$l_1$};
              \draw(3.7,-3.5)node{$\mathcal{Z}_1$};
              \draw(1,-5)node{$r_2$};
              \draw(3.5,-6)node{$l_2$};
          

              \draw[<->](2,0.5)--(4.5,0.5);
              \draw(0,0.5)node[anchor=south]{$\mathcal A$};
              \draw(2,0.5)node[anchor=south]{$\mathcal B$};
              \draw(4.5,0.5)node[anchor=south]{$\mathcal C$};

              \draw(3.25,0.5)node[anchor=south]{$d_1$};
              \draw[<->](0,0.5)--(2,0.5);
              \draw(1,0.5)node[anchor=south]{$d_2$};
              \draw[<->](-0.5,-7.5)--(-0.5,-3.5);
              \draw(-0.5,-4.75)node[anchor=east]{$2d_2$};
              \draw[<->](-0.5,-2)--(-0.5,0);
              \draw(-0.5,-1)node[anchor=east]{$d_2$};

              \draw[<->](5,-8.5)--(5,-3.5);
              \draw(5,-5.75)node[anchor=west]{$2d_1$};
              \draw[<->](5,-3)--(5,0);
              \draw(5,-1.5)node[anchor=west]{$d_1$};

              \draw[<->,dashed](t0p1)--(t2p1);
              \draw(2,-1.75)node[anchor=west]{$d'$};
            \end{tikzpicture}
          \end{center}

          The distances imposed by the signals allow to conclude:
          \begin{eqnarray}
            d_1+d_2&=&2d_2+d'-d_1\label{eqn1}\\
            2d_1&=&d_2+d'\nonumber
          \end{eqnarray}
          \begin{eqnarray}
            d_2+d_1&=&2d_1+d'-d_2\label{eqn2}\\
            2d_2&=&d_1+d'\nonumber
          \end{eqnarray}
          
          Equation (\ref{eqn1}) comes from the $r_1$ signal from the leftmost
          bottom $\#$, and equation (\ref{eqn2}) comes from the rightmost bottom
          one. From this we can deduce that $d_1=d_2=d'$ and that the $R_1/r_1$
          signal emitted at time $t_2$ by the central column reaches the right
          column at time $t_0$, contradicting the fact that the state at time $t_0$
          in this right column is $\#'$.
        \item It was created by $l_1/r_1$, and so was the right $\#$ at $t_1$.
          In this case, one of them has to have another $\#$ in its past, by
          Lemma~\ref{columns-lemma}. Thus, this new $\#$ sends a $r_2$ or $l_2$
          signal (depending on which one we consider), that necessarily collides
          either with the other column of $\#'$s, or to the $L_1/R_1$ signals
          represented in thick on the figure below. Since this $L_1$ signal
          creates the top left $\#$, this collision should occur necessarily
          before time $t_0$ and would create a $\kappa$ which is forbidden by
          hypothesis.
          
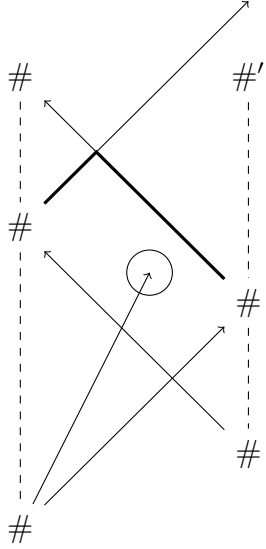
\begin{center}
            \begin{tikzpicture}
              \draw(0,0)node(t0p0)[fill=white]{$\#$};
              \draw(3,0)node(t0p1)[fill=white]{$\#'$};

              \draw(0,-2)node(t2p0)[fill=white]{$\#$};
              \draw(3,-3)node(t1p1)[fill=white]{$\#$};

              \draw(3,-5)node(t3p1)[fill=white]{$\#$};

              \draw[->](t1p1)--(t0p0);
              \draw[->](t3p1)--(t2p0);

              \draw[->](t2p0)--(3,1);

              \draw(0,-6)node(t4p0)[fill=white]{$\#$};
              \draw[->](t4p0)--(t1p1);

              \draw[very thick](t2p0)--(1,-1)--(t1p1);
              \draw[->](t4p0)--(1.7,-2.6);
              \draw(1.7,-2.6)circle(0.3cm);

              \draw[dashed](t0p0)--(t2p0)--(t4p0)
              (t0p1)--(t1p1)--(t3p1);
            \end{tikzpicture}
          \end{center}
      \end{itemize}

    \item Or it was created by $l_2/r_2$. In this case, the only possible option
      for the $\#$ at $t_2$ on the left column is a $l_1/r_1$, else there would
      be another $\#/\#'$ column between the $\#$ and the $\#'$ at time $t_0$,
      contrarily to our hypothesis.

      Therefore we are in the exact symmetric of the case studied above, where
      we supposed that the $\#$ at $t_2$ was created by $l_2/r_2$, and the left
      $\#$ at $t_1$ by $l_1/r_1$. We can conclude to a contradiction by the same
      reasoning.
    \end{itemize}

  \end{proof}

\section{Proof of Lemma  ~\ref{l2r2-prop}}

  \begin{proof}
    We want to use the above Lemma~\ref{l1r1-prop} to prove this one. To do this, we
    need to prove the existence of a configuration of the form $\#\Sigma^*\#'$
    or its symmetric $\#'\Sigma^*\#$, with in any case the $\#$ created by $l_1/r_1$
    signals, in the past of our current configuration.

    We begin by proving the existence of another column of $\#$, on the left
    of the left $\#$. Since the right column of $\#'$ collides with $\mathcal Z$ --
    the $l_2$ signal creating the left $\#$, there is necessarily a $\#$ in the past
    of this $\#'$, say at time $t_1$. But then, this $\#$ is necessarily the one who
    emitted $\mathcal Z$, for else it would have been formed by $l_2/r_2$ signals colliding
    with $\mathcal Z$, thus forming another $\#$ between the $\#$ and the $\#'$ of our
    hypothesis. This shows that $\mathcal Z$ was originated by
    a $\#$ in the same column than the right $\#'$.

    But then, this $\#$ must have emitted also an $L_1$ signal, that cannot collide
    with the $r_2$ signal (let us call it ${\mathcal Z}'$),
    colliding with $\mathcal Z$ to form our left $\#$. Thus, this
    signal must be converted into an $l_1$ before meeting the $r_2$ of the hypothesis;
    more precisely, it must collide with an $R_1$ signal. But this one necessarily has
    a finite origin, since it could not have crossed ${\mathcal Z}'$. This shows the
    existence of another column of $\#^{(')}$, at $t_2$.

    
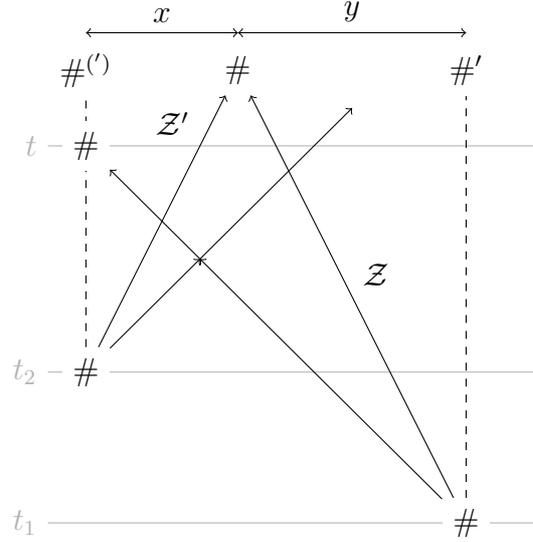
\begin{center}
      \begin{tikzpicture}

        \draw[black!30!white](-2.5,-1)node[anchor=east]{$t$}--(4,-1)
        (-2.5,-4)node[anchor=east]{$t_2$}--(4,-4)
        (-2.5,-6)node[anchor=east]{$t_1$}--(4,-6);

        \draw(0,0)node[fill=white](t0p0){$\#$}
        (3,0)node[fill=white](t0p1){$\#'$}
        (3,-6)node[fill=white](t1p1){$\#$};
        \draw(-2,-4)node[fill=white](t2p2){$\#$};

        \draw[dashed](t1p1)--(t0p1)
        (t2p2)--(-2,0)node[fill=white]{$\#^{(')}$};

        \draw[->](t1p1)--(t0p0);
        \draw(1.5,-3)node[anchor=south west]{$\mathcal Z$};
        \draw(-0.5,-1)node[anchor=south east]{${\mathcal Z}'$};
        \draw[->](t1p1)--(-0.5,-2.5);
        \draw(-2,-1)node(t3p2)[fill=white]{$\#$};
        \draw[->](-0.5,-2.5)--(t3p2);

        \draw[->](t2p2)--(-0.5,-2.5);
        \draw[->](-0.5,-2.5)--(1.5,-0.5);
        \draw[->](t2p2)--(t0p0);

        \draw[<->](-2,0.5)--(0,0.5);
        \draw[<->](0,0.5)--(3,0.5);
        \draw(-1,0.5)node[anchor=south]{$x$}
        (1.5,0.5)node[anchor=south]{$y$};

      \end{tikzpicture}
    \end{center}
      But then, since both $\#$s at $t_1$ and $t_2$ emit a $L_1$ or $R_1$ signal,
      and they are desynchronized, one of this signal has to hit the other column
      before $t_0$. In fact, to avoid this, we should have at the same time
      $x+y>2y$ and $x+y>2x$, which is impossible. Thus, we can apply Proposition
     ~\ref{l1r1-prop} or its symmetric to this step (at time $t$ on the figure).
  \end{proof}

\section{Proof of Lemma  ~\ref{periodic-prop}}

  \begin{proof}
    We first need to show that in a configuration $x\in X_S$ containing $\#$s,
    there is no history diagram yielding $x$ in which two consecutive $\#$s of $x$
    were both created by $r_2$/$l_2$ signals, or both by $r_1$/$l_1$ signals.
    \begin{itemize}
      \item if two consecutive $\#$s were created by $r_2$/$l_2$ signals, these
        signals would have collided before, thus yielding another $\#$ between
        them.
      \item if two consecutive $\#$s were created by $r_1$/$l_1$, the right $r_1$
        would originate in the left column, and vice-versa:
        
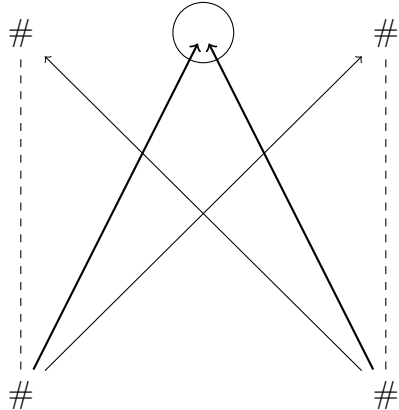
\begin{center}
          \begin{tikzpicture}[scale=0.8]
            \draw(0,0)node[fill=white](t0p0){$\#$}
            (3,0)node[fill=none,draw=none](t0p1){}
            (6,0)node[fill=white](t0p2){$\#$};
            
            \draw(3,0)circle(0.5cm);
            
            \draw(0,-6)node[fill=white](t1p0){$\#$}
            (6,-6)node[fill=white](t1p2){$\#$};
            
            \draw[dashed](t1p0)--(t0p0)(t1p2)--(t0p2);
            
            \draw[->](t1p0)--(t0p2);
            \draw[->,thick](t1p0)--(t0p1);
            \draw[->,thick](t1p2)--(t0p1);
            \draw[->](t1p2)--(t0p0);
          \end{tikzpicture}
        \end{center}
        These signals would then have been created by $\#$s, which
        also would have emitted $r_2$s and $l_2$s intersecting between the two
        $\#$s in question, thus creating another $\#$.
    \end{itemize}

    Now if we have two consecutive $\#$s in a configuration, we can infer
    the position of a third one from the signals having created them, since they 
    have different speeds and thus must intersect. Then, we can reproduce this
    argument at the step where these signals were created.
    
    The only problem with this
    is that the infered $\#$ could be on the same side at each step of the argument,
    thus yielding a configuration with a rightmost $\#$ or a leftmost $\#$. Without
    loss of generality, let us assume there is rightmost $\#$. Then, the past of this cell
    would be an infinite column of $\#$s/$\#'$s, with infinitely many $\#$s. But then,
    some $r_2$ signal emitted by a $\#$ of this column would collide with an $R_1$ emitted
    later, which would yield a $K$ state: the configuration we considered would not be in
    $X_S$.

    Then, it is easy to see that there can be no signal between two consecutive $\#$s,
    for they would come from a desynchronized configuration in a history diagram of
    the configuration.
  \end{proof}

\section{Proof of Lemma  ~\ref{validlanguage}}

  \begin{proof}
    Applying Lemma~\ref{periodic-prop} at time $t$, we now that the
    configuration at time $t$ is of the form
    $^\omega(\#B^n)^\omega$. Moreover, using again the argumentation
    of Lemma~\ref{periodic-prop}, we know that their is some time $t'$
    before $t$ in history where the configuration is of the form
    $^\omega(\#B^{n'})^\omega$ with $n'>n$. Hence, $L$ is exactly the
    language of forward orbits of periodic configurations of the form
    $^\omega(\#B^n)^\omega$. It is straightforward to check from the
    definition of $S$ (see Figure~\ref{fig:deffire}) that the language
    of periods of such forward orbits is a finite union of languages
    of the form:
    \[\#'B^{x_1}A_1\cdots B^{x_i}A_i\] where $i\leq 4$, $A_i\in\Sigma$ and where the
  $x_i$ must satisfy some simple linear equations. The lemma follows.
  \end{proof}

\section{Proof of Lemma  ~\ref{one-ptime}}

  \begin{proof}
    There are several possibilities for the number of $\#$ in the past of the
    $\#$ of $\#'$. There may be at most two, for else, two of them would have been
    created by $l_1$ signals crossing the $R_1$/$r_1$ emitted by the third $\#$,
    and that intersection, circled in the following picture, can occur only once.
    
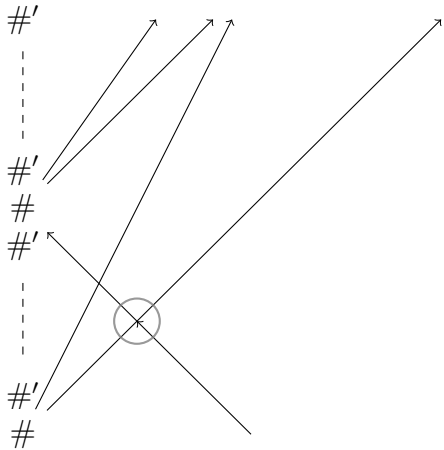
\begin{center}
      \begin{tikzpicture}
        \draw[dashed](0,0)node[fill=white,draw=none](t0){$\#'$}--
        (0,-2)node[fill=white](t11){$\#'$}
        (0,-3)node[draw=none,fill=white,circle](t13){$\#'$}
        (0,-2.5)node[draw=none,circle](t12){$\#$};
        \draw[dashed](t13)--(0,-5)node[draw=none,circle,fill=white]{$\#'$}
        (0,-5.5)node[draw=none,fill=white](t2){$\#$};
        \draw[->](t2)--(5.5,0);
        \draw[->](t2)--(2.75,0);
        \draw[->](t12)--(2.5,0);
        \draw[->](t12)--(1.75,0);
        
        \draw[->](3,-5.5)--(1.5,-4);
        \draw[->](1.5,-4)--(t12);
        \draw[gray,thick,opacity=0.8](1.5,-4)circle(0.3cm);
      \end{tikzpicture}
    \end{center}
    Thus, we are left with:

    \begin{itemize}
    \item Two $\#$ in the past: in this case, the $r_2$ they emitted allows to
      guess the whole column.  After the last $r_2$ signal, there may be signals
      in the opposite direction, that cannot collide in the past, i.e.  first
      $l_2$s, then $L_1$s and $l_1$s. Such configurations are described, with
      $x=2z+y$, by:
        $$\{R_1,r_1,B\}^\omega\{r_2,B\}^*l_1B^xl_2B^yL_1B^zl_2B^z\{\#',\#\}
        B^zr_2B^zR_1B^yr_2B^xr_1\{L_1,l_1,B\}^\omega$$

      \item With only one $\#$ in the past, there may still be some $l_1$ signal
        somewhere:
        
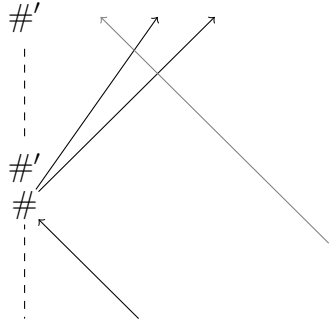
\begin{center}
          \begin{tikzpicture}
            \draw[dashed](0,0)node[fill=white,draw=none](t0){$\#'$}--
            (0,-2)node[fill=white](t11){$\#'$}
            (0,-2.5)node[draw=none,fill=white,inner sep=0,circle](t12){$\#$}--(0,-4);

            \draw[->](t12)--(2.5,0);
            \draw[->](t12)--(1.75,0);
            
            \draw[->](1.5,-4)--(t12);
            \draw[->,gray](4,-3)--(1,0);
          \end{tikzpicture}
        \end{center}
        Thus the configurations are of the following form:
        $$\{R_1,r_1,B\}^\omega L_1B^zl_2B^z\{\#',\#\}
        B^zr_2B^zR_1\{L_1,l_1,B\}^\omega~,$$
        with possibly a $r_1$ replacing some state on the left,
        and an $l_1$ on the right, like on the figure.

      \item With no $\#$ in the past, we have an infinite column of $\#'$, 
        thus no signal can come from it. These configurations are of the form:
        $$^\omega\{R_1,r_1,B\}\{r_2,B\}^*\#'\{l_2,B\}^*\{L_1,l_1,B\}^\omega$$
    \end{itemize}
    The lemma easily follows from the characterization of the different possible
    forms of configuration discussed above.
  \end{proof}

\section{Proof of Lemma   \ref{neutral-sharpsprime}}
  \begin{proof}
    In a history diagram satisfying the hypothesis their can be no collision that would generate a $\#$ and even no signal crossing (because this would mean that two signal are going in opposite directions in the same part of a configuration and thus that at least one would meet a $\#'$ in the past, which is forbidden by the hypothesis).     Therefore we have: $L=\bigl(B^+\{R_1,r_1\}\bigr)^\ast\bigl(B^+ r_2\bigr)^\ast\bigl(B^+\#'\bigr)^\ast\bigl(B^+ l_2\bigr)^\ast\bigl(B^+ \{L_1,l_1\}\bigr)^\ast$.
  \end{proof}

\section{Proof of Lemma \ref{zero-regular}}
  \begin{proof}
  Without any $\#$ nor $\#'$, we are left with blank states and signals.
  The configurations of the limit set are those configurations in which the
  possible signals do not have to collide in the past, thus fall in one of these
  three possibilities:
  $$^\omega\{l_1,r_2,B\}^\omega$$
  $$^\omega\{r_1,l_2,B\}^\omega$$
  $$^\omega\{r_1,R_1,B\}\{r_2,B\}^*\{l_2,B\}^*\{l_1,L_1,B\}^\omega$$
  \end{proof}

\end{document}